\documentclass[12pt]{article}

\usepackage[utf8]{inputenc}
\usepackage[english]{babel}
\usepackage[margin=1in]{geometry}
\usepackage{graphicx, epstopdf, rotating}

\usepackage{amssymb, amsmath, amsthm, bm, siunitx, natbib, color, anyfontsize, rotating, adjustbox}
\newcommand{\R}{\mathbb{R}}    % Real numbers
\newcommand{\Nd}{\mathcal{N}}  % Normal distribution
\newcommand{\bx}{\bm{x}}       % bold x
\newcommand{\by}{\bm{y}}       % bold y
\newcommand{\bu}{\bm{u}}       % bold u
\newcommand{\bs}{\bm{s}}       % bold s
\newcommand{\bmu}{\bm{\mu}}    % bold mu
\newcommand{\bth}{\bm{\theta}} % bold theta
\newcommand{\bS}{\bm{\Sigma}}  % bold Sigma
\newcommand{\ind}[1]{\bm{1}_{\set{#1}}}    % Indicator function
\newcommand{\norm}[2][]{\left\Vert#2\right\Vert_{#1}} % norm (\norm[p]{\bx} for p norm)
\newcommand{\set}[1]{\left\{ #1 \right\}}             % set notation
\newcommand{\inp}[1]{\left\langle #1 \right\rangle}   % inner product
\newtheorem{theorem}{Theorem}
\newtheorem*{theorem*}{Theorem}

\newtheorem*{lemma*}{Lemma}
\newtheorem{corollary}{Corollary}
\newtheorem*{corollary*}{Corollary}

\newtheorem*{definition*}{Definition}

\numberwithin{equation}{section}

\title{Flexible nonstationary spatio-temporal modeling of high-frequency monitoring data}

\author{Christopher J. Geoga\thanks{corresponding author: \texttt{christopher.geoga@rutgers.edu}}
  \thanks{Department of Statistics, Rutgers University}
  \thanks{Division of Mathematics and Computer Science, Argonne National Laboratory}
        \and
        Mihai Anitescu \footnotemark[3]
                       \thanks{Department of Statistics, University of Chicago}
        \and
        Michael L. Stein \footnotemark[2] \footnotemark[4]
      }
\date{}

\begin{document}

\maketitle

\begin{abstract}
  Many physical datasets are generated by collections of instruments that make
measurements at regular time intervals. For such regular monitoring data, we
extend the framework of half-spectral covariance functions to the case of
nonstationarity in space and time and demonstrate that this method provides a
natural and tractable way to incorporate complex behaviors into a covariance
model. Further, we use this method with fully time-domain computations to obtain
bona fide maximum likelihood estimators---as opposed to using Whittle-type likelihood
approximations, for example---that can still be computed efficiently. We apply
this method to very high-frequency Doppler LIDAR vertical wind velocity
measurements, demonstrating that the model can expressively capture the extreme
nonstationarity of dynamics above and below the atmospheric boundary layer and,
more importantly, the interaction of the process dynamics across it.

\end{abstract}

\textbf{Keywords:}
Spectral domain, 
Doppler LIDAR,
Meteorology,
Atmospheric Boundary Layer

\bigskip

\section{Introduction}
Gaussian process (GP) models are versatile and useful tools for
studying spatial and spatio-temporal processes. A GP model is completely
specified by postulating a mean and covariance function, so that a collection of
data points $\by = \set{y_j}_{j=1}^N$ observed at locations
$\set{\bx_j}_{j=1}^N$ has the distribution
\begin{equation*} \label{eq:GPmodel}
  \by \sim \Nd(\bmu, \bS),
\end{equation*}
where $\bmu(\bx)$ is a mean function and $\bS_{j,k} = K(\bx_j, \bx_k)$ is the
covariance matrix induced by covariance function $K$. A common estimation
problem is to select a parametric form for $K$ and attempt to recover the
parameters from data. For the duration of this paper, we assume $\bmu \equiv 0$.
While this may seem restrictive, the analysis of anomaly fields instead of raw data
is commonplace, and extensions of these methods to a nonzero mean would not
be overly difficult. 

If $K(\bx, \by)$ is a function only of $\bx-\by$, we call the process model
\emph{stationary}, an assumption which is enormously useful for a variety of
both theoretical and computational reasons.  For the vast majority of real-world
processes, however, even among the ones for which a Gaussian process assumption
is tolerable, a stationarity assumption is not. With modern computing tools and
instrument infrastructure, spatial datasets are commonly large enough and dense
enough that even visual inspection of the data will make it abundantly clear
that a stationary process model would be fundamentally misrepresentative of the
dynamics of the process. For this reason, developing natural and flexible
nonstationary models to more accurately capture dependence structure in modern
datasets is an exigent statistical problem.

Modeling such nonstationarity is difficult, however, in no small part because it
is challenging to specify valid covariance models via positive definite
functions. In the stationary case, perhaps the easiest way to verify that a
function $K_{\bth}$ is a valid covariance function is to find an integrable
function for which $K$ in its one-argument form $K_{\bth}(\bx, \by) \equiv
K_{\bth}(\bx-\by)$ is its Fourier transform. One can then express $K_{\bth}$ as
\begin{equation*} %\label{eq:}
  K_{\bth}(\bs) = \int e^{i 2 \pi \inp{\bm{f}, \bs}} S_{\bth}(\bm{f})
  \mathrm{d} \bm{f},
\end{equation*}
where $S_{\bth}(\bm{f})$ is called the \emph{spectral density} corresponding to
$K$.  If $S_{\bth}(\bm{f})$ is positive and symmetric about zero, Bochner's
theorem implies that $K_{\bth}$ is a positive definite real-valued function and
thus a valid covariance function. In general, it is easier to specify positive
and symmetric functions than positive definite ones, and many of the more exotic
stationary covariance functions are derived, or at least confirmed to be valid,
in the spectral domain. Many theoretical results about GPs that are relevant to
the parameter estimation problem, like those about the equivalence and orthogonality
of Gaussian measures, are proven and applied in the spectral domain
\citep{ibragimov1978, stein1999, zhang2004, zhang2005}, further emphasizing that
the spectral domain is arguably often a more productive way to approach studying
covariance structure.

Unfortunately, nonstationarity implies the nonexistence of a spectral density,
and the full-dimensional Fourier transform of most covariance functions only
exists in the distributional sense. Accordingly, it is necessary to find other
ways to specify valid models, and there is a wealth of literature in the field
that contributes methods and ideas towards solving this problem. Examples
include careful stitching together of valid ``local'' models \citep{higdon2002,
stein2005, paciorek2006}, deforming the coordinate system of a stationary
process \citep{genton2004, anderes2008}, spatial mixtures of different
stationary fields \citep{fuentes2001b}, nonparametric representations
\citep{sampson1992, jun2008}, multiresolution models \citep{nychka2002,
matsuo2011, katzfuss2017}, Bayesian methods \citep{katzfuss2012}, separable
models \citep{genton2007}, and many others. But each method has drawbacks, and
for some it is difficult to compute true maximum likelihood estimators, so that
approximated likelihoods are commonly substituted.  With this in mind, we extend
a methodology for continuing to think about nonstationary covariances in the
spectral domain by modeling marginal spectra and coherences. Like every other
method listed here, this strategy has its limitations and is not always
available or prudent to use. It does, however, provide a way to incorporate and
potentially compose several different ways of specifying nonstationarity and, as
will be shown here, can be pushed into complex forms while maintaining validity,
numerical robustness, and expressiveness.

\section{Half-spectral covariance functions}

A key observation made in \cite{cressie1999} and \cite{stein2005b} is that one
can obtain a valid model by specifying a ``half-spectral'' form of a covariance
function for a random field $Z(t, \bx)$, whereupon the stationary space-time
covariance function is obtained by performing an inverse Fourier transform of a
function $h(f, \bx-\bx')$ with respect to time. The function
\begin{equation} \label{eq:sit_original}
  K(t-t', \bx-\bx') := \int_{\Omega} e^{2 \pi i f (t-t')} h(f, \bx-\bx')
  \mathrm{d} f
\end{equation}
gives a valid covariance function for discrete-time data ($\Omega = [-1/2\Delta
t, 1/2 \Delta t)$) or continuous-time data ($\Omega=\R$). In discrete time, we
advocate using $\sin(\pi f)$ instead of $f$ so that the periodic extension of
the model is smooth with respect to $f$ at the endpoints.  \cite{stein2005b}
proposes an expression for $h$ of the form 
\begin{equation*} %\label{eq:}
  h(f, \bx-\bx') = S(f) C_{f}(\bx-\bx') e^{i g(f) \bu^T (\bx-\bx')}
\end{equation*}
for a stationary but non-separable and asymmetric model, where $S(f)$ can be
interpreted to be a marginal spectral density in time; $C_{f}$---a valid spatial
correlation function---is the modulus of the coherence of $Z(\cdot, \bx)$ and
$Z(\cdot, \bx')$; and $g(f)\bu^T(\bx-\bx')$ is the phase relation. Varying
$C_f$ in (\ref{eq:sit_original}) allows flexibility in the nature of
nonseparability, with separability being the special case of $C_f$ and $g(f)$
not depending on $f$. 

For measurements made at regular time intervals and a relatively small number of
spatial locations, this modeling perspective can be considered as an augmented
multiple time series approach for which incorporating new time series
(measurements at a new spatial location) is possible in a model-consistent way
and without the requirement of any new parameters. It affords the practitioner a
convenient and approachable method for modeling space-time dependence primarily
by thinking about marginal---or at least lower-dimensional---dependence
structure. If the data is regular in time, for example, one can visually inspect
marginal time spectra and immediately include features of those estimates in a
valid, fully spatio-temporal covariance model.  See \cite{horrell2017} for an
extension to more complex stationary models of this form and \cite{guinness2013}
for a nonstationary model in a similar vein.

We now extend this modeling framework to the case of processes that are
nonstationary in space but stationary in time. Because this modeling perspective
relies so heavily on specifying marginals, a nonstationary extension is very
direct.
\begin{theorem} \label{thm:thm1}
Let $\set{S_{\bx}(f)}$ be temporal spectral densities indexed by spatial
locations $\bx$. If  $C_f(\bx, \bx')$ is a real-valued correlation function and
$g(f)$ is an odd function, then
\begin{equation} \label{eq:sit}
  K((t, \bx), (t', \bx')) = 
  \int_\Omega e^{i 2 \pi f (t-t') + i g(f)\bu^T(\bx-\bx')} 
  \sqrt{S_{\bx}(f) S_{\bx'}(f)}
C_f(\bx, \bx') \mathrm{d} f
\end{equation}
defines a valid covariance function for a spatially nonstationary real-valued
space-time process $Z(t, \bx)$.
\end{theorem}
\begin{proof}
  As in the proof provided in \cite{stein2005b}, we consider the matrix-valued function
\begin{equation} \label{eq:Phi}
  \bm{\Phi}(f) := \set{
    \sqrt{S_{\bx_j} (f) S_{\bx_k} (f)} C_f (\bx_j, \bx_k) e^{i g(f) \bu^T (\bx_j -
    \bx_k')}
  }_{j,k =1}^{n}.
\end{equation}  
It is sufficient to prove that $\bm{\Phi}(f)$ is positive semi-definite for each
$f$ and integrable. Note that we may write $\bm{\Phi}(f) = \bm{\mathcal{S}}(f) \circ
\bm{C}(f) \circ \mathcal{P}(f)$, a Hadamard product of the time-spectra term,
correlation function, and phase term.  $\bm{C}(f)$ is by definition positive
definite, and we note that $\bm{\mathcal{S}}(f)$ is positive semi-definite by
recognizing the rank one structure of $\bm{\mathcal{S}}(f) = \bu(f) \bu(f)^T$ with
$\bm{u}_j(f) = \sqrt{S_{\bx_j}(f)}$. Similarly, $\mathcal{P}(f)$ has a Hermitian
rank-one symmetric factorization since $e^{i g(f) \bu^T (\bx-\bx')} = e^{i g(f)
\bu^T \bx} e^{-i g(f) \bu^T \bx'}$.  By the Schur product theorem, then, we
conclude that $\bm{\Phi}(f)$ is positive semi-definite for all $f$. Since $\sqrt{S_{\bx_j} (f)
S_{\bx_k} (f)} \leq \max_{j} S_{\bx_j}(f)$ is integrable and both $C_f$ and the
complex exponential are bounded, $\bm{\Phi}(f)$ is integrable. Finally, we
observe that $\bm{\Phi}(f) = \overline{\bm{\Phi}(-f)}$ and thus conclude that K
is a real-valued positive definite function.
\end{proof}

Several comments are in order. First, the proof of Theorem \ref{thm:thm1}
demonstrates that this argument extends to a much broader class of functions
than the one that will be studied here. For example, an algebraic representation
of $\bm{\mathcal{S}}(f) = \bm{Q}(f) \bm{A}(f) \bm{Q}(f)^T$ where $\bm{A}$ is
positive semidefinite could provide an even richer nonparametric
representation for marginal and cross spectra of a spatio-temporal process.
Building rank-deficient positive semi-definite matrices is easy, and per the
argument of the proof above, such an extension would also yield an equally valid
real-valued covariance function. 

Second, temporal stationarity---or at most a restrictive form of
nonstationarity---is required by this model. Most importantly, the existence of
a spectral density is an implication of stationarity along that dimension. It is
of course possible to have processes that are locally stationary but whose
spectra evolve slowly and continuously with time (as were studied in
\cite{poppick2016}, for example), but even this extension would substantially
complicate the numerics required to build covariance matrices for maximum
likelihood. Further, as will be discussed below, performing space-time domain
computations with this model with minimal numerical concerns is only made
feasible by the use of the Fast Fourier Transform algorithm \citep{cooley1965}.
This means that the covariance function can only conveniently be evaluated at
the inverse Fourier frequencies, and so other relaxations of temporal
stationarity like deformation-type methods would also be difficult to apply.

Finally, we note that the above argument still applies if $C_f$ is a bounded
covariance function and not a correlation function, and that requiring it to be
a correlation function primarily serves to mitigate identifiability problems.
Moreover, $C_f$ need not be stationary. As will be demonstrated later in this
work, between the space-dependent marginal time spectra $S_{\bx}$ and a
nonstationary coherence function $C_f(\bx, \bx')$, it is easily possible to
compose this model with other nonstationary models to obtain an even more
flexible valid covariance function. 

Despite the difficulty discussed above, there are several restrictive ways to
relax the temporal stationary of this model. While slightly more complex and
involved ideas are possible, in this work we only employ the simplest possible
extension in the form of a nonstationary scale.
\begin{corollary}
  If $K$ is the function given in (\ref{eq:sit}) and $\lambda(\bx, t) \geq 0$
  everywhere, then the function
  \begin{equation} \label{eq:st_sit}
    \tilde{K}( (\bx, t), (\bx', t') ) := \lambda(\bx, t) \lambda(\bx', t') 
    K( (\bx, t), (\bx', t') ) + \eta^2 \ind{(\bx, t) = (\bx', t')}
  \end{equation}
  is also a valid covariance function.
\end{corollary}
With the further addition of a ``temporal nugget'' $\eta_{t}^2 \ind{t=t'}$, in
the same spirit as the ``spatial nugget'' of \cite{gneiting2002}, this
model is precisely the formulation we use here to explore our data.

For a relatively small number of spatial locations, using the FFT to numerically
evaluate these integrals for all pairs of spatial locations and then to assemble
covariance matrices is straightforward. With little numerical work, then, one
can write simple models for marginal---or at least lower-dimensional---behavior
and obtain all the function values necessary to build a valid space-time
covariance matrix in $\mathcal{O}(n_s^2 n_t \log_2 n_t)$ time complexity, where
$n_s$ is the number of unique spatial locations and $n_t$ is the time width of
the process.  In particular, the integral approximation for discrete-time
processes given by
\begin{equation*} %\label{eq:}
  \int_{-1/2}^{1/2} 
  e^{2 \pi i f t} g(\sin \pi f) df 
  \approx 
  \frac{1}{N} \sum_{j=0}^{N-1} e^{ 2 \pi i f_j t} g(\sin \pi f_j)
\end{equation*}
for Fourier frequencies $\set{f_j}$ can be evaluated simultaneously and
efficiently at all necessary time lags $t$ via the inverse FFT algorithm. This is
effectively an application of trapezoidal integration, and by also taking
sufficiently long FFTs to avoid edge effects, these approximations can easily be
made accurate to double precision for kernels that can be checked with a closed
form expression. In our experience, taking an FFT length of $N = 7 \cdot n_t$ is
more than sufficient. To the degree that any linear algebra for large
ill-conditioned matrices is exact in finite precision arithmetic, then, the
evaluation of the likelihood with covariance matrices assembled from this
approximation is exact. As further evidence of this, the log-likelihood for FFT
sizes ranging from $5 \cdot n_t$ to $21 \cdot n_t$ all agree to at least single
precision in our experimentation, further suggesting that this numerical scheme
provides effectively exact function evaluation and thus provides a way to
minimize the exact likelihood.

Despite the many FFTs being computed here, in our use case with on the order of
$30$ spatial locations and one thousand time locations, evaluating the function
at all necessary space-time locations and aggregating the values takes less than
one second on a Dell Latitude E5470 with an Intel Core i5-6200U@2.3GHz CPU, a
mobile processor not designed for heavy computation, and poses no numerical or
computational challenge. Once these values (and the values for the derivatives
of (\ref{eq:st_sit}) with respect to kernel parameters) have been computed,
covariance matrices can be assembled and the likelihood and its derivatives can
be evaluated as usual.

To motivate the application of this model, let us introduce the dataset to be
explored in this work.

\section{Doppler LIDAR wind measurement data}

The US Department of Energy's Atmospheric Radiation Measurement (ARM) program
was created with the goal to provide long-term in situ and remotely sensed
observations in various climate regimes to improve the understanding of
processes impacting atmospheric radiation \citep{stokes1994}. The Southern Great
Plains (SGP) observatory is the first ARM site and the largest climate research
facility in the world, equipped with more than 50 instruments and providing
continuous measurements ranging from basic meteorology and radiation to cloud
and aerosol properties at several locations in the north-central Oklahoma and
south Kansas region \citep{mather2013, sisterson2016}. In this work, we study
measurements of the vertical component of the wind field at very high
spatio-temporal resolution as observed by the $\SI{1.5}{\micro\meter}$ pulsed
Doppler lidar (DL) deployed at the SGP central facility (CF).\footnote{The ARM
processed/quality-controlled DL fixed-beam stare measurements are freely
available \citep{Newsom2010}. The raw data was used in this work, which can be
obtained upon request at \texttt{https://www.arm.gov} using the ``Ask us'' page.}

The principle of operation of the DL is similar to that of the Doppler radars in
that they both transmit pulses of energy into the atmosphere and measure the
returned signal scattered back by clear-air irregularities \citep{gage1978,
muradyan2020}. In other words, scattering is assumed to originate from
atmospheric particulates moving at the same speed as the wind. The ARM DL has a
full upper-hemispheric scanning capability, measuring the one-dimensional
velocity projections across a range of angles in order to fully resolve the
three-dimensional wind field \citep{newsom2012}. The general DL scanning
strategy for this purpose at the SGP CF is to perform 8-beam
plan-position-indicator scans once every approximately 15 minutes, and all
sampling between those times is done in a vertically staring mode. When the DL
is pointing vertically, it provides height- and time-resolved measurements of
the vertical velocity. 

As we focus exclusively on vertical measurements here, the resulting data from
this instrument is multiple gappy time series of the vertical wind velocity
components at $\SI{30}{\meter}$ vertical intervals (commonly referred to as ``range
gates"), which extend to a maximum of $\SI{9.6}{\km}$ range and are sampled at
approximately $1.2$ second intervals.  For more details about the instrument's
theory of operation, functionality, and configuration, see \cite{newsom2012}.
Figure \ref{fig:data_tsplot} shows an example collection of time series plots of
the data at discrete range gates over the duration of fifteen minutes on June 2,
2015, while Figure \ref{fig:data_hmplot} shows the time-height cross-section of
the same data from about $\SI{200}{\meter}$ to $\SI{2.25}{\km}$ above ground level (AGL). 

\begin{figure}[!ht]
  \centering
  \input{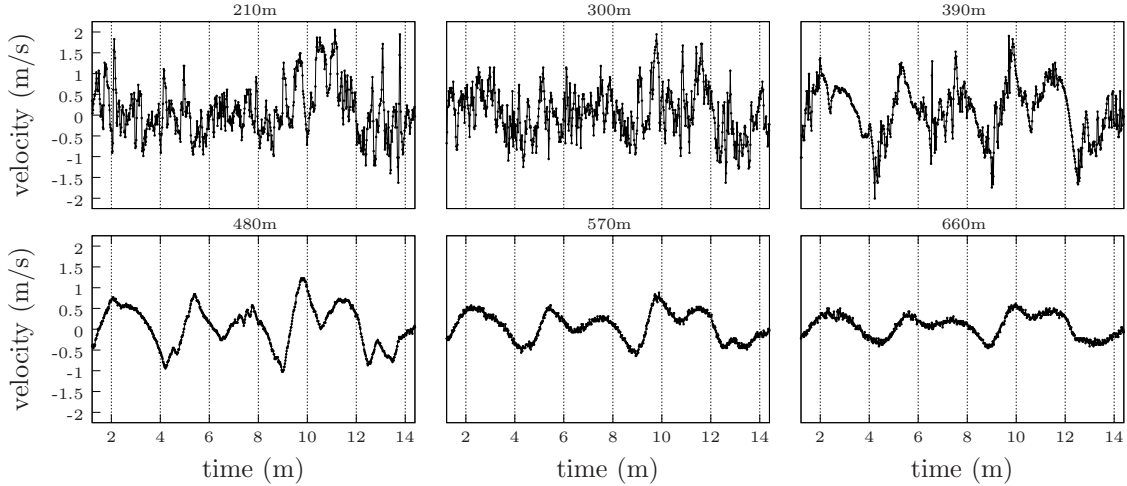}
  \caption{Fifteen minutes of Doppler LIDAR vertical wind velocity measurements
  at several different range gates on June $2$, $2015$ at $1400$UTC,
  shortly after sunrise.} 
  \label{fig:data_tsplot}
\end{figure}

As one would expect, this data exemplifies extreme nonstationarity, which no
sensible model could ignore. In space, the primary source of nonstationarity is
the height of the atmospheric boundary layer (ABL), which is the layer directly
in contact with the earth's surface through turbulent mixing processes
\citep{tucker2009}. As can be seen in Figure \ref{fig:data_hmplot} below, the
ABL height increases during the day via solar heating of the surface, resulting
in a well-mixed daytime convective boundary layer (BL). In the absence of solar
radiation at night, the convection is stopped by heat loss to space, resulting
in a thin nocturnal BL with typically
$\SIrange[range-phrase=-,range-units=single]{100}{300}{\meter}$ depth
\citep{cushman2014}.  Above the BL height, the concentrations of particulates
such as aerosols and cloud particles, which are the source of scattering,
dramatically decrease, resulting in a substantial reduction in the backscatter
signal and thus measurement quality.  Typically, only the lowest $2$ to
$\SI{3}{\km}$ of the atmosphere yields high-quality measurements
\citep{newsom2012}.  In addition, the minimum range for the DL is approximately
$\SI{100}{\meter}$, and even if data is available below that height, the lowest
range gates are suspect and are recommended not to be used (personal
communicaton with instrument mentor). As the maximum operational range of the
ARM DL is up to $\SI{9.6}{\km}$, this is a substantial truncation, but
inspection of the data on its full domain would involve a great deal of
preprocessing measurement artifacts, which is outside the scope of this work.\footnote{\texttt{\texttt{https://plot.dmf.arm.gov/PLOTS/sgp/sgpdlfpt/20150602/sgpdlfptC1.b1.fpt\_12-15hour.20150602.png}},
  for example, gives an example of ARM's automatic visualization of DL
measurements from $1200$ to $1500$ UTC on June $02$, $2015$, demonstrating the
vertical extent of measurements with good SNR.}

In order to study the dependence structure across the BL and capture the
dynamics above and below simultaneously without having to contend too deeply
with measurement quality concerns or computational issues, we focus this work on
the time $1400$ UTC---shortly after the local sunrise in northern
Oklahoma---when the ABL height has grown to approximately
$\SIrange[range-phrase=-,range-units=single]{400}{500}{\meter}$ on
average.  This part of the day is convenient as the ABL height is well above the
minimum altitude of DL operation, so that plenty of high-quality measurements
are available below it as well as above it. For the purpose of manageable data
sizes, we focus on altitudes between about $200$ and $\SI{850}{\meter}$ in order to
have large spatial samples of range gates both above and below the BL while also
being able to work with full measurement segments between 3D scans. For
reference, considering range gates $7$ to $28$ and a full time segment of
vertical measurements (about $14$ minutes) amounts to approximately $17$
thousand data points.

\begin{figure}[!ht]
  \centering
  \makebox[\textwidth][c]{\input{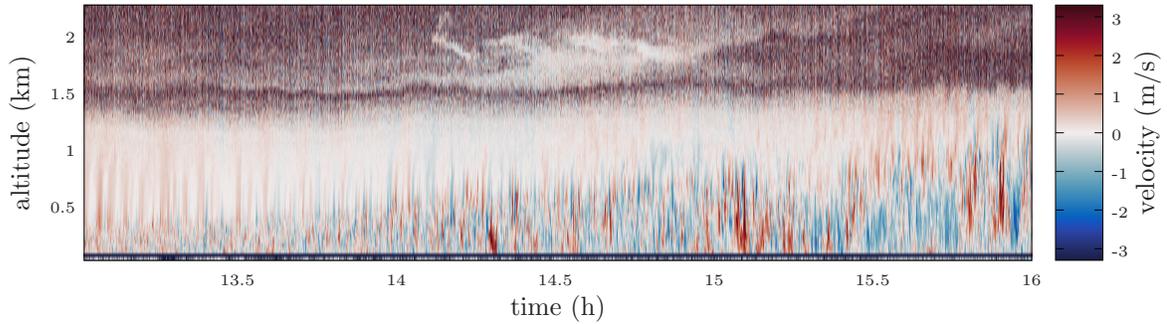}}
  \caption{An alternative view of the process from $1300$ to $1600$ UTC on June
    $24$th, $2015$, a day with particularly high activity, demonstrating the
    rising mixing height after sunrise and the dynamic domain of data
    reliability. The color scale for this image is highly truncated, as the
  unreliable measurements at altitudes with low SNR range sporadically between
  $\SI[per-mode=symbol]{20}{\meter\per\second}$ and
  $\SI[per-mode=symbol]{-20}{\meter\per\second}$.} 
  \label{fig:data_hmplot}
\end{figure}

This data source has several features and properties that make it interesting to
study. For one, the spatial component is vertical and not horizontal. Just as
time should not be treated as just another spatial dimension, we believe that
the vertical dimension of many physical processes should be treated with the
same care. The measurements here provide a good testing set for vertical
modeling as well as a strong justification of this claim, as the spatial
nonstationarity is very sharp and likely not even once differentiable, which is
unusual to observe on such a fine spatial scale and with such high covariance.

As well as being uncommon in its spatial axis, this data source is a good motivating
example for nonstationary models. The rate at which marginal dynamics change is
noticeably faster than the rate at which covariance information decays, which is
to say that two range gates can have obviously different dynamics but still be
strongly informative about each other. The measurements at $300$ and
$\SI{480}{\meter}$ in Figure \ref{fig:data_tsplot}, for example, clearly have
very strong coherence at low frequencies despite how different finer-scale
behavior is at each altitude.  Any locally stationary model would necessarily be
throwing out a great deal of information of that kind, which is arguably of the
highest scientific interest. In some cases, the creation of intricate
nonstationary models has been criticized for not actually capturing more
behavior, as it is frequently the case that most information that can be learned
about a spatial or spatio-temporal process corresponds to high frequency
behavior of the spectral density (or is otherwise ``micro-ergodic''), which thus
largely corresponds to local information \citep{stein1999}. This data is an
example of a situation when that type of pessimism is unnecessary. 

Finally, the measurements used in this work are made at a much higher frequency
than most other sources of environmental data. At the time scales of these
measurements, a great deal can be inferred about smoothness properties of the
field, as well as the study of wind gusts, turbulence and other small-scale
meteorological phenomena, which will especially be critical in the lower portion
of the ABL. In general, the features that require high-frequency resolution to
observe also make the data more challenging to study. Things like measurement
quantization when the wind velocity is very low at night, for example, are real
concerns, as well as distinguishing real short time-scale velocity fluctuations
from instrument noise. As will be clear, this work is only scratching the
surface of what can be studied with this type of high-frequency and
high-resolution monitoring data, and it is precisely its associated challenges that make
it interesting.

\section{A half-spectral model for Doppler LIDAR data}

One of the strongest features of the data, as can be seen in Figure
\ref{fig:data_hmplot}, is the strong low-frequency coherence across the entire
domain, but especially above the ABL. Another is that the scale of the
measurements appears to decay as the altitude increases above the ABL height.
For the purposes of both parameter economy and interpretability, we introduce
first a functional form that we use extensively in several parts of this
model. Taking inspiration from the frequency response function of the
Butterworth filter from signal processing \citep{butterworth1930}, whose
signature feature is being flat in some neighborhood of zero and then decaying
algebraically, we define the utility function $B(z)$ as
\begin{equation*} \label{eq:Sx}
  B(z; \xi_0, \xi_1, \xi_2) 
  := e^{\xi_0} \left[1+ \left( \frac{z}{\xi_1} \right)^{2 \xi_2} \right]^{-1}.
\end{equation*}
As we demonstrate below, this function is convenient in three different parts of
the specification of the model given by Equation \ref{eq:st_sit}.

With that preliminary complete, let us begin by specifying a form for the
nonstationary scale parameters, which we choose to represent in a flexible form
as
\begin{equation*} %\label{eq:}
  \lambda(x,t) = B( (x-\beta)_+ \; ; \; 0, \phi_1, \phi_2)
    \sum_{j=1}^4 w_j(t) \theta_{t_j},
\end{equation*}
where $\set{t_j}$ are knots in time chosen evenly across the domain, and
$w_j(t)$ are time-dependent weights that are normalized to sum to one, with
unnormalized forms given by $\tilde{w}_j(t) := e^{-|t-t_j|/50 \Delta t}$ (for
  reference, the time window for the data that will ultimately be fitted here is
about $775 \Delta t$). The spatial dependence of the scale comes in the form of
the Butterworth-type function $B$, which starts to decay once the altitude $x$,
in units of meters, is above the ABL height, which we denote with $\beta$, and is
itself a model parameter.  Due to the structure we see across several days of
data, we prefer this representation for the spatial dependence on scale to
knot-based function approximation.

To introduce spatial nonstationarity of components in the marginal spectrum and
coherence while maintaining interpretability, we implement a very simple
space-dependent parameter mapping given by, for example with the
parameter $\rho$,
\begin{equation*} %\label{eq:}
  \rho(x) := (1-w(x)) \rho_{0} + w(x) \rho_{1},
\end{equation*}
where
\begin{equation} \label{eq:parm_interp}
  w(x) := (1+e^{-\tau (x - \beta)})^{-1},
\end{equation}
and the subscripts $0$ and $1$ denote in this case the regimes of being below
and above the ABL respectively, so that effectively the spatially-indexed
parameters are logistic interpolants between two values.  The inverse scale
parameter $\tau$ of this logistic function is informative about the sharpness of
the transition from below-ABL and above-ABL dynamics, and is
very important to this model. For one reason, fixing a constant ABL height
$\beta$ for $15$ minutes is obviously not correct from a brief inspection of
Figure \ref{fig:data_hmplot}. Since that is a form of temporal nonstationarity
that this model cannot easily account for, we introduce this $\tau$ parameter to
at least allow the model the ability to communicate about the sharpness of this
transition and capture better the dynamics at range gates near the mean height
of the ABL.

Equipped with this parametric form of spatial dependence, we define the marginal
spectral densities with 
\begin{equation} \label{eq:sdf} 
  S_x(f) = (1+B(\sin \pi f \; ; \; \xi_0(x), \xi_1, \xi_2)) (e^{\rho(x)} \sin^2 \pi f + 1)^{-\nu(x) - 1/2},
\end{equation}
where the spatial dependence of the low frequency multiplier $1+B$ 
comes in the form of the parameter $\xi_0(x)$, and $\xi_1$ and $\xi_2$ are fixed
across space. As a small technical note, for time series measured at non-unit
time steps, a change of variables in the Fourier transform suggests that a
$\Delta t$ rescaling should be applied to $S_x$. We do not do that here and
instead elect to simply treat the time sampling as unit sampling and scale
nonparametric spectral density estimators consistently with that in the
following section, implicitly absorbing this $\Delta t \approx 1.2$s into the
$\theta_j$ parameters. 

This spectral density corresponds closely to a filtered Mat\'ern spectral
density for continuously indexed processes, which is best characterized by its
algebraic decay, so that $S(\bm{f}) \sim \norm[2]{\bm{f}}^{-2 \nu - d}$ in $d$
dimensions.  In discrete time, $\nu$ is not so directly interpretable since even
negative values yield spectral densities corresponding to non-pathological
processes. But the algebraic tail decay is nonetheless convenient and at least
heuristically informative about smoothness that would be implied by a continuous
version of the model, and using $\sin \pi f$ instead of $f$ extends this
functional form nicely to a function that is smooth on the unit torus. The extra
multiplier of $1+B$ provides the model the ability to add extra low frequency
energy without affecting the higher frequencies, which is a pronounced feature
of the data and can be seen in the diagnostic figures of the next section.

The coherence function $C_f$ is crucial to this model, and is extremely
different above and below the ABL. To be as flexible as possible, we choose
$C_f$ to be a nonstationary Paciorek-Schervish-type correlation function with an
extension to variable smoothness \citep{stein2005} given by
\begin{equation} \label{eq:Cf}
  C_f(x,x') := \frac{\sqrt[4]{\gamma(x) \gamma(x')}}{\gamma(x,x')} 
  \mathcal{M}_{(\nu_s(x)+\nu_s(x'))/2} \left( \frac{|x-x'|}{\gamma(x,x')} \right),
\end{equation}
where $\gamma(x,x') := \sqrt{(\gamma(x) + \gamma(x))/2}$, and we parameterize 
\begin{equation*} %\label{eq:}
  \gamma(x)(f) := B(\sin \pi f; \; \zeta_0(x), \zeta_1(x), \zeta_2(x)),
\end{equation*}
using $\zeta$ instead of $\xi$ to distinguish between the parameters of $B(f)$
for $S_x$ and $C_f$. Unfortunately the spatially indexed smoothnesses for $C_f$
are not smoothly interpolated like other spatially indexed parameters in order
to avoid taking derivatives of second-kind Bessel functions (refer to the
Discussion section for more information) in the derivatives of $C_f$, and are
instead discontinuously parameterized as
\begin{equation*} %\label{eq:}
  \nu_s(x) = \mathbf{1}\set{x \leq \beta}\frac{1}{2} + (1-\mathbf{1}\set{x \leq
  \beta})\frac{3}{4}.
\end{equation*}
We acknowledge that it is somewhat noteworthy to have fixed smoothnesses for
this correlation function. Considering that there are only approximately $20$ spatial
locations being considered in this model, however, our experimentation indicates
that the effect of varying these parameters on the likelihood is modest.

Since the true phase dynamics of this process are almost certainly driven by
nonlinearity that no Gaussian process model could capture, we include only the
simplest form of $g(f)$ in this model, setting $g(f) := \alpha \sin \pi f$, with
no spatial variation. With the inclusion of our space-time and temporal nuggets,
we now have a complete covariance function of the form (\ref{eq:st_sit}).  Table
$1$ provides a final summary and reference for the interpretation of each
parameter in the model.

\begin{table}[!ht]
  \centering
  \begin{tabular}{|c|c|}
    \hline
    $\beta$&Boundary layer height\\
$\tau$&Shape parameter for boundary layer regime transition rate\\
$\theta_{t_j}$&Local scale parameters\\
$\phi_{j}$&Scale decay with altitude over boundary layer\\
$\xi_j(x)$&Space-dependent shape parameters for temporal low-frequency multipler $1+B(f)$\\
$\rho(x)$&Space-dependent temporal range-type parameter\\
$\nu(x)$&Space-dependent temporal smoothness-type parameter\\
$\zeta_j(x)$&Space-dependent shape parameters for coherence function $\gamma(x)$\\
$\alpha$&Phase function rate-type parameter\\
$\eta_{st}$&Spatio-temporal nugget size\\
$\eta_{t}$&Purely temporal nugget size\\
 
    \hline
  \end{tabular}
  \label{tab:parms_summary}
  \caption{A summary of all the parameters used in the covariance model as
  parameterized in (\ref{eq:st_sit}).}
\end{table}

With this model fully introduced, it merits comment that since it will be fitted
with maximum likelihood estimators, parameters informing the high frequency
behavior will be highly priotized over parameters that inform low frequency
behavior. Because of the human tendency to focus on lower frequency behavior
when performing visual inspection of simulations or other diagnostics, however,
we have included several parameters and degrees of freedom here that do not
heavily affect the likelihood. The $\set{\xi(x)}$ parameters, for example, are a
good example of extra degrees of freedom that do not strongly affect the
likelihood and are completely confined to the very lowest of frequencies. Their
effect on visual diagnostics for conditional simulations, however, are very
noticeable.  For certain purposes, a smaller model could be used with similarly
satisfying results, and the supplemental information compares visual diagnostics
of the full model shown here and a reduced one. Since this work is an
investigation into how well a Gaussian process model can capture the very
complex dynamics of vertical wind velocity, however, we proceed with a
maximalist approach here.

\section{Estimation Results}

Using the methods described at the end of section $2$, it is possible to
evaluate the likelihood and its derivatives entirely in the space-time domain.
In this section, we discuss the results of maximum likelihood estimation for the
above model to six different days of data in June $2015$ that were judged to
have comparable meteorological conditions. In each case, we performed the
optimization over the $25$ parameters using a hand-written trust region method
adapted directly from \cite{nocedal2006}. In order to only use first derivatives
of the covariance function with respect to parameters, we used a symmetrized
stochastic estimator for the expected Fisher information matrix
\citep{geoga2019} as an approximation for the Hessian (observed information). To
further avoid costly matrix-matrix operations, we also used a symmetrized stochastic
gradient.  Optimization was terminated at a relative precision of $10^{-6}$ in
the objective function. 

To enable unconstrained optimization, the model was fitted with transformed
variables, as many need to be positive. This was necessary due to the stochastic
derivative estimators, as several parameters being reasonably close to zero
means that even a very accurate stochastic gradient and expected Fisher matrix
can cause domain issues.  All point estimates and standard deviations provided in
this section, however, are provided for the model as parameterized above,
obtained by un-transforming the computed MLE and evaluating the likelihood and
its derivatives once in that domain.  This is not ideal considering that
high-quality convergence criteria such as the size of the Newton step change in
response to a nonlinear transformation of the variables, but we do not believe
this issue to be overly invalidating in the interpretation of the estimates or
their associated uncertainty.  The code used to perform the estimation is spread
out over the software packages
\texttt{HalfSpectral.jl}\footnote{\texttt{https://git.sr.ht/\textasciitilde
cgeoga/HalfSpectral.jl}} for the half-spectral kernel function and
\texttt{GPMaxlik.jl}\footnote{\texttt{https://git.sr.ht/\textasciitilde
cgeoga/GPMaxlik.jl}} for efficient optimization of the likelihood, both written
in the Julia programming language \citep{bezanson2017}.  Example files in
\texttt{HalfSpectral.jl} contain some specific scripts used to obtain the
estimates below.

As described in Section $3$, we restrict the estimation here to approximately
$14$ minute time segments (about $780$ measurements) and range gates between
$\SI{210}{\meter}$ (the seventh gate) and either $\SI{690}{\meter}$ (the $23$rd
gate) or $\SI{840}{\meter}$ (the $28$th gate) based on the height of the ABL so
that both above- and below-mixing height dynamics have reasonable
representation.  These time segements correspond to measurements between the
gaps noted previously, and the resulting data sizes are near the limit of the
computational abilities of the hardware used for the estimation.  We do stress,
however, that the methods described here could equally well be applied to gappy
regular monitoring data by virtue of the fact that performing the computations
in the time domain means that an FFT is being applied to the analytical
spectrum, not the data itself.

\begin{table}[!ht] 
 \centering
 \begin{adjustbox}{width=\textwidth}
 \begin{tabular}{|c|c|c|c|c|c|c|}
   \hline
   &02&03&06&20&24&28\\

   \hline
   $\theta_{t_0}$&\num[scientific-notation=false,round-mode=places,round-precision=2]{1.0556262087602715} (\num[scientific-notation=false,round-mode=places,round-precision=2]{0.02642429876657453})&\num[scientific-notation=false,round-mode=places,round-precision=2]{0.7229962513558726} (\num[scientific-notation=false,round-mode=places,round-precision=2]{0.08892856293646752})&\num[scientific-notation=false,round-mode=places,round-precision=2]{1.0301866449348027} (\num[scientific-notation=false,round-mode=places,round-precision=2]{0.03298906851974212})&\num[scientific-notation=false,round-mode=places,round-precision=2]{1.1271223866528852} (\num[scientific-notation=false,round-mode=places,round-precision=2]{0.02309972054915364})&\num[scientific-notation=false,round-mode=places,round-precision=2]{1.6509842401806383} (\num[scientific-notation=false,round-mode=places,round-precision=2]{0.02840315863216571})&\num[scientific-notation=false,round-mode=places,round-precision=2]{1.5435982654353244} (\num[scientific-notation=false,round-mode=places,round-precision=2]{0.03623436809530996})\\
$\theta_{t_2}$&\num[scientific-notation=false,round-mode=places,round-precision=2]{1.1816083418260148} (\num[scientific-notation=false,round-mode=places,round-precision=2]{0.027966037586357417})&\num[scientific-notation=false,round-mode=places,round-precision=2]{0.6169231208893362} (\num[scientific-notation=false,round-mode=places,round-precision=2]{0.08960036456210184})&\num[scientific-notation=false,round-mode=places,round-precision=2]{1.2259183038185775} (\num[scientific-notation=false,round-mode=places,round-precision=2]{0.034777977360913565})&\num[scientific-notation=false,round-mode=places,round-precision=2]{0.9177684976167811} (\num[scientific-notation=false,round-mode=places,round-precision=2]{0.025601825905434982})&\num[scientific-notation=false,round-mode=places,round-precision=2]{1.7803151762076286} (\num[scientific-notation=false,round-mode=places,round-precision=2]{0.029463409312946054})&\num[scientific-notation=false,round-mode=places,round-precision=2]{1.2326691445678453} (\num[scientific-notation=false,round-mode=places,round-precision=2]{0.03818865778752037})\\
$\theta_{t_3}$&\num[scientific-notation=false,round-mode=places,round-precision=2]{1.2660682932086842} (\num[scientific-notation=false,round-mode=places,round-precision=2]{0.027916400371029864})&\num[scientific-notation=false,round-mode=places,round-precision=2]{0.6355976308042246} (\num[scientific-notation=false,round-mode=places,round-precision=2]{0.089550349423295})&\num[scientific-notation=false,round-mode=places,round-precision=2]{1.6267502035230752} (\num[scientific-notation=false,round-mode=places,round-precision=2]{0.03438179985239311})&\num[scientific-notation=false,round-mode=places,round-precision=2]{1.1613792516530297} (\num[scientific-notation=false,round-mode=places,round-precision=2]{0.024741167106438887})&\num[scientific-notation=false,round-mode=places,round-precision=2]{1.8587772575635308} (\num[scientific-notation=false,round-mode=places,round-precision=2]{0.029517967936912044})&\num[scientific-notation=false,round-mode=places,round-precision=2]{1.2523659501669144} (\num[scientific-notation=false,round-mode=places,round-precision=2]{0.03791652900510134})\\
$\theta_{t_4}$&\num[scientific-notation=false,round-mode=places,round-precision=2]{1.1896106548426282} (\num[scientific-notation=false,round-mode=places,round-precision=2]{0.026263241978222388})&\num[scientific-notation=false,round-mode=places,round-precision=2]{0.8245483472049486} (\num[scientific-notation=false,round-mode=places,round-precision=2]{0.08891335566182942})&\num[scientific-notation=false,round-mode=places,round-precision=2]{1.9798337925416563} (\num[scientific-notation=false,round-mode=places,round-precision=2]{0.03261413678690081})&\num[scientific-notation=false,round-mode=places,round-precision=2]{1.2640619496159615} (\num[scientific-notation=false,round-mode=places,round-precision=2]{0.02295996838594376})&\num[scientific-notation=false,round-mode=places,round-precision=2]{2.2641401237086787} (\num[scientific-notation=false,round-mode=places,round-precision=2]{0.028220857609009838})&\num[scientific-notation=false,round-mode=places,round-precision=2]{0.8741260287419821} (\num[scientific-notation=false,round-mode=places,round-precision=2]{0.03680445030038951})\\
$\rho_{0}$&\num[scientific-notation=false,round-mode=places,round-precision=2]{2.7184535523514923} (\num[scientific-notation=false,round-mode=places,round-precision=2]{0.14064359526562117})&\num[scientific-notation=false,round-mode=places,round-precision=2]{0.3523791436292875} (\num[scientific-notation=false,round-mode=places,round-precision=2]{0.4968336055967104})&\num[scientific-notation=false,round-mode=places,round-precision=2]{3.4769355585614514} (\num[scientific-notation=false,round-mode=places,round-precision=2]{0.14465230003381357})&\num[scientific-notation=false,round-mode=places,round-precision=2]{2.368616665065095} (\num[scientific-notation=false,round-mode=places,round-precision=2]{0.11420809480531673})&\num[scientific-notation=false,round-mode=places,round-precision=2]{2.7576042937339356} (\num[scientific-notation=false,round-mode=places,round-precision=2]{0.13367068150340966})&\num[scientific-notation=false,round-mode=places,round-precision=2]{2.863413789112529} (\num[scientific-notation=false,round-mode=places,round-precision=2]{0.15714562596727327})\\
$\nu_{0}$&\num[scientific-notation=false,round-mode=places,round-precision=2]{1.1436376841231368} (\num[scientific-notation=false,round-mode=places,round-precision=2]{0.058287478028003215})&\num[scientific-notation=false,round-mode=places,round-precision=2]{7.0710367001363315} (\num[scientific-notation=false,round-mode=places,round-precision=2]{2.534591305764318})&\num[scientific-notation=false,round-mode=places,round-precision=2]{1.0230752171559516} (\num[scientific-notation=false,round-mode=places,round-precision=2]{0.0418826624881223})&\num[scientific-notation=false,round-mode=places,round-precision=2]{1.427851074192237} (\num[scientific-notation=false,round-mode=places,round-precision=2]{0.06564508178499827})&\num[scientific-notation=false,round-mode=places,round-precision=2]{1.2108339448267966} (\num[scientific-notation=false,round-mode=places,round-precision=2]{0.054117649305205906})&\num[scientific-notation=false,round-mode=places,round-precision=2]{1.345209228510582} (\num[scientific-notation=false,round-mode=places,round-precision=2]{0.06659087531108375})\\
$\rho_{1}$&\num[scientific-notation=false,round-mode=places,round-precision=2]{5.525591539265126} (\num[scientific-notation=false,round-mode=places,round-precision=2]{0.1753826549617934})&\num[scientific-notation=false,round-mode=places,round-precision=2]{2.3107594378019876} (\num[scientific-notation=false,round-mode=places,round-precision=2]{0.5460793137635274})&\num[scientific-notation=false,round-mode=places,round-precision=2]{6.76116370045953} (\num[scientific-notation=false,round-mode=places,round-precision=2]{0.11844637355716209})&\num[scientific-notation=false,round-mode=places,round-precision=2]{8.926909979121648} (\num[scientific-notation=false,round-mode=places,round-precision=2]{0.15126267957360842})&\num[scientific-notation=false,round-mode=places,round-precision=2]{4.885356713012435} (\num[scientific-notation=false,round-mode=places,round-precision=2]{0.17997826670026837})&\num[scientific-notation=false,round-mode=places,round-precision=2]{6.2469586032345426} (\num[scientific-notation=false,round-mode=places,round-precision=2]{0.18848628469733272})\\
$\nu_{1}$&\num[scientific-notation=false,round-mode=places,round-precision=2]{2.139742354992843} (\num[scientific-notation=false,round-mode=places,round-precision=2]{0.17539554669050203})&\num[scientific-notation=false,round-mode=places,round-precision=2]{20.052700039213416} (\num[scientific-notation=false,round-mode=places,round-precision=2]{9.774254200268746})&\num[scientific-notation=false,round-mode=places,round-precision=2]{3.1365648035091573} (\num[scientific-notation=false,round-mode=places,round-precision=2]{0.18744343747793027})&\num[scientific-notation=false,round-mode=places,round-precision=2]{0.371209470709341} (\num[scientific-notation=false,round-mode=places,round-precision=2]{0.025830281901005607})&\num[scientific-notation=false,round-mode=places,round-precision=2]{2.113313519968596} (\num[scientific-notation=false,round-mode=places,round-precision=2]{0.167712388596908})&\num[scientific-notation=false,round-mode=places,round-precision=2]{2.3513540244417315} (\num[scientific-notation=false,round-mode=places,round-precision=2]{0.24894162883358184})\\
$\zeta_{00}$&\num[scientific-notation=false,round-mode=places,round-precision=2]{12.45648190493277} (\num[scientific-notation=false,round-mode=places,round-precision=2]{0.4946364957453865})&\num[scientific-notation=false,round-mode=places,round-precision=2]{13.637335735835578} (\num[scientific-notation=false,round-mode=places,round-precision=2]{0.1771803409234416})&\num[scientific-notation=false,round-mode=places,round-precision=2]{0.496812859113053} (\num[scientific-notation=true,round-mode=figures,round-precision=2,output-exponent-marker=\texttt{e}]{27107.97053909355})&\num[scientific-notation=false,round-mode=places,round-precision=2]{17.465524606902807} (\num[scientific-notation=false,round-mode=places,round-precision=2]{0.8860365223431688})&\num[scientific-notation=false,round-mode=places,round-precision=2]{11.890452625368212} (\num[scientific-notation=false,round-mode=places,round-precision=2]{0.29379470534132784})&\num[scientific-notation=false,round-mode=places,round-precision=2]{11.482418807490845} (\num[scientific-notation=false,round-mode=places,round-precision=2]{0.44248021560214035})\\
$\zeta_{01}$&\num[scientific-notation=false,round-mode=places,round-precision=2]{0.01652013151928776} (\num[scientific-notation=true,round-mode=figures,round-precision=2,output-exponent-marker=\texttt{e}]{0.005620330799788851})&\num[scientific-notation=false,round-mode=places,round-precision=2]{0.0104849924240422} (\num[scientific-notation=true,round-mode=figures,round-precision=2,output-exponent-marker=\texttt{e}]{0.0020984283893744618})&\num[scientific-notation=false,round-mode=places,round-precision=2]{0.03318469941523771} (\num[scientific-notation=true,round-mode=figures,round-precision=2,output-exponent-marker=\texttt{e}]{0.005245187155527688})&\num[scientific-notation=true,round-mode=figures,round-precision=2,output-exponent-marker=\texttt{e}]{0.00020853824085970704} (\num[scientific-notation=true,round-mode=figures,round-precision=2,output-exponent-marker=\texttt{e}]{0.00012030071780367899})&\num[scientific-notation=false,round-mode=places,round-precision=2]{0.07039749552931138} (\num[scientific-notation=false,round-mode=places,round-precision=2]{0.012665293648571704})&\num[scientific-notation=false,round-mode=places,round-precision=2]{0.020008578723901934} (\num[scientific-notation=true,round-mode=figures,round-precision=2,output-exponent-marker=\texttt{e}]{0.0076550916354132725})\\
$\zeta_{02}$&\num[scientific-notation=false,round-mode=places,round-precision=2]{0.8395641019094884} (\num[scientific-notation=false,round-mode=places,round-precision=2]{0.04638496983658248})&\num[scientific-notation=false,round-mode=places,round-precision=2]{0.7714550752618892} (\num[scientific-notation=false,round-mode=places,round-precision=2]{0.04384598242088745})&\num[scientific-notation=false,round-mode=places,round-precision=2]{1.353570734601728} (\num[scientific-notation=false,round-mode=places,round-precision=2]{0.06032062781466546})&\num[scientific-notation=false,round-mode=places,round-precision=2]{0.6817845001137893} (\num[scientific-notation=false,round-mode=places,round-precision=2]{0.034942528971098764})&\num[scientific-notation=false,round-mode=places,round-precision=2]{1.0883851973298733} (\num[scientific-notation=false,round-mode=places,round-precision=2]{0.05552315972402215})&\num[scientific-notation=false,round-mode=places,round-precision=2]{0.7266489288773986} (\num[scientific-notation=false,round-mode=places,round-precision=2]{0.05300958731624841})\\
$\zeta_{10}$&\num[scientific-notation=false,round-mode=places,round-precision=2]{16.79317692040071} (\num[scientific-notation=false,round-mode=places,round-precision=2]{0.5022660109612367})&\num[scientific-notation=false,round-mode=places,round-precision=2]{13.568906732409545} (\num[scientific-notation=false,round-mode=places,round-precision=2]{0.21384306855975357})&\num[scientific-notation=false,round-mode=places,round-precision=2]{18.133182962067444} (\num[scientific-notation=false,round-mode=places,round-precision=2]{0.32919779453613024})&\num[scientific-notation=false,round-mode=places,round-precision=2]{24.81462524005045} (\num[scientific-notation=false,round-mode=places,round-precision=2]{0.843653431181234})&\num[scientific-notation=false,round-mode=places,round-precision=2]{14.952590939136945} (\num[scientific-notation=false,round-mode=places,round-precision=2]{0.3074962905524487})&\num[scientific-notation=false,round-mode=places,round-precision=2]{17.25380228624976} (\num[scientific-notation=false,round-mode=places,round-precision=2]{0.4912372546023238})\\
$\zeta_{11}$&\num[scientific-notation=false,round-mode=places,round-precision=2]{0.011595522882194436} (\num[scientific-notation=true,round-mode=figures,round-precision=2,output-exponent-marker=\texttt{e}]{0.0025110517435301774})&\num[scientific-notation=false,round-mode=places,round-precision=2]{10.657394685525706} (\num[scientific-notation=false,round-mode=places,round-precision=2]{3.524882488558071})&\num[scientific-notation=true,round-mode=figures,round-precision=2,output-exponent-marker=\texttt{e}]{0.0005314875468409665} (\num[scientific-notation=true,round-mode=figures,round-precision=2,output-exponent-marker=\texttt{e}]{0.0001933401239973053})&\num[scientific-notation=true,round-mode=figures,round-precision=2,output-exponent-marker=\texttt{e}]{1.5011577752937066e-5} (\num[scientific-notation=true,round-mode=figures,round-precision=2,output-exponent-marker=\texttt{e}]{1.0028947183795867e-5})&\num[scientific-notation=false,round-mode=places,round-precision=2]{0.0391165513171305} (\num[scientific-notation=true,round-mode=figures,round-precision=2,output-exponent-marker=\texttt{e}]{0.004538829409167363})&\num[scientific-notation=true,round-mode=figures,round-precision=2,output-exponent-marker=\texttt{e}]{0.004385388344358817} (\num[scientific-notation=true,round-mode=figures,round-precision=2,output-exponent-marker=\texttt{e}]{0.001090989900762981})\\
$\zeta_{12}$&\num[scientific-notation=false,round-mode=places,round-precision=2]{1.383669202874107} (\num[scientific-notation=false,round-mode=places,round-precision=2]{0.09423185167686503})&\num[scientific-notation=false,round-mode=places,round-precision=2]{69.02732185122191} (\num[scientific-notation=false,round-mode=places,round-precision=2]{23.142467575824526})&\num[scientific-notation=false,round-mode=places,round-precision=2]{0.6742817767766122} (\num[scientific-notation=false,round-mode=places,round-precision=2]{0.05692442191104085})&\num[scientific-notation=false,round-mode=places,round-precision=2]{1.0226053055396143} (\num[scientific-notation=false,round-mode=places,round-precision=2]{0.04863554493527654})&\num[scientific-notation=false,round-mode=places,round-precision=2]{1.8611592274086597} (\num[scientific-notation=false,round-mode=places,round-precision=2]{0.12400207566148859})&\num[scientific-notation=false,round-mode=places,round-precision=2]{0.970383928988044} (\num[scientific-notation=false,round-mode=places,round-precision=2]{0.0981887348454128})\\
$\beta$&\num[scientific-notation=true,round-mode=figures,round-precision=2,output-exponent-marker=\texttt{e}]{458.841071692587} (\num[scientific-notation=false,round-mode=places,round-precision=2]{2.147502500896736})&\num[scientific-notation=true,round-mode=figures,round-precision=2,output-exponent-marker=\texttt{e}]{659.9881297939746} (\num[scientific-notation=false,round-mode=places,round-precision=2]{4.0990711316415736})&\num[scientific-notation=true,round-mode=figures,round-precision=2,output-exponent-marker=\texttt{e}]{472.4214122571234} (\num[scientific-notation=false,round-mode=places,round-precision=2]{3.181003466565874})&\num[scientific-notation=true,round-mode=figures,round-precision=2,output-exponent-marker=\texttt{e}]{682.399434990027} (\num[scientific-notation=false,round-mode=places,round-precision=2]{1.2329858538476153})&\num[scientific-notation=true,round-mode=figures,round-precision=2,output-exponent-marker=\texttt{e}]{583.8088691710553} (\num[scientific-notation=false,round-mode=places,round-precision=2]{3.4932477947914147})&\num[scientific-notation=true,round-mode=figures,round-precision=2,output-exponent-marker=\texttt{e}]{545.7345172635958} (\num[scientific-notation=false,round-mode=places,round-precision=2]{3.9448430801624337})\\
$\tau$&\num[scientific-notation=false,round-mode=places,round-precision=2]{0.05645016419419213} (\num[scientific-notation=true,round-mode=figures,round-precision=2,output-exponent-marker=\texttt{e}]{0.0020722032977364972})&\num[scientific-notation=false,round-mode=places,round-precision=2]{0.04072698290322612} (\num[scientific-notation=true,round-mode=figures,round-precision=2,output-exponent-marker=\texttt{e}]{0.002356610706978164})&\num[scientific-notation=false,round-mode=places,round-precision=2]{0.02390157373609535} (\num[scientific-notation=true,round-mode=figures,round-precision=2,output-exponent-marker=\texttt{e}]{0.0007762165635489626})&\num[scientific-notation=false,round-mode=places,round-precision=2]{0.02704652196385024} (\num[scientific-notation=true,round-mode=figures,round-precision=2,output-exponent-marker=\texttt{e}]{0.0004470900057160451})&\num[scientific-notation=false,round-mode=places,round-precision=2]{0.03330645678601953} (\num[scientific-notation=true,round-mode=figures,round-precision=2,output-exponent-marker=\texttt{e}]{0.0013159757438945697})&\num[scientific-notation=false,round-mode=places,round-precision=2]{0.03604004978089501} (\num[scientific-notation=true,round-mode=figures,round-precision=2,output-exponent-marker=\texttt{e}]{0.0012775597470298748})\\
$\xi_{00}$&\num[scientific-notation=false,round-mode=places,round-precision=2]{3.578639178296149} (\num[scientific-notation=false,round-mode=places,round-precision=2]{0.8557458067026765})&\num[scientific-notation=false,round-mode=places,round-precision=2]{18.982493267434105} (\num[scientific-notation=false,round-mode=places,round-precision=2]{8.801716889295792})&\num[scientific-notation=false,round-mode=places,round-precision=2]{1.0229125613814773} (\num[scientific-notation=false,round-mode=places,round-precision=2]{0.3887020405885409})&\num[scientific-notation=false,round-mode=places,round-precision=2]{14.940267793718695} (\num[scientific-notation=false,round-mode=places,round-precision=2]{5.071835072114729})&\num[scientific-notation=false,round-mode=places,round-precision=2]{0.4681169435795914} (\num[scientific-notation=false,round-mode=places,round-precision=2]{0.22311036151810143})&\num[scientific-notation=false,round-mode=places,round-precision=2]{37.06966971908233} (\num[scientific-notation=false,round-mode=places,round-precision=2]{22.654546983234997})\\
$\xi_{01}$&\num[scientific-notation=false,round-mode=places,round-precision=2]{7.076909630417635} (\num[scientific-notation=false,round-mode=places,round-precision=2]{1.918324190944924})&\num[scientific-notation=false,round-mode=places,round-precision=2]{9.213539692471493} (\num[scientific-notation=false,round-mode=places,round-precision=2]{4.724728718719655})&\num[scientific-notation=false,round-mode=places,round-precision=2]{0.02006093596881308} (\num[scientific-notation=false,round-mode=places,round-precision=2]{0.23371434288264523})&\num[scientific-notation=false,round-mode=places,round-precision=2]{38.47810754602705} (\num[scientific-notation=false,round-mode=places,round-precision=2]{14.238775076964256})&\num[scientific-notation=false,round-mode=places,round-precision=2]{1.5744938910769006} (\num[scientific-notation=false,round-mode=places,round-precision=2]{0.6097192570126274})&\num[scientific-notation=false,round-mode=places,round-precision=2]{38.305892677745945} (\num[scientific-notation=false,round-mode=places,round-precision=2]{24.499053550449357})\\
$\xi_2$&\num[scientific-notation=false,round-mode=places,round-precision=2]{6.552778376515583} (\num[scientific-notation=false,round-mode=places,round-precision=2]{2.9303540737137266})&\num[scientific-notation=false,round-mode=places,round-precision=2]{0.8519164107207683} (\num[scientific-notation=false,round-mode=places,round-precision=2]{0.1311410352988271})&\num[scientific-notation=false,round-mode=places,round-precision=2]{12.982330294317986} (\num[scientific-notation=false,round-mode=places,round-precision=2]{36.07055887142275})&\num[scientific-notation=false,round-mode=places,round-precision=2]{1.0905482297519487} (\num[scientific-notation=false,round-mode=places,round-precision=2]{0.18257162503432992})&\num[scientific-notation=false,round-mode=places,round-precision=2]{5.987846877884677} (\num[scientific-notation=false,round-mode=places,round-precision=2]{4.896466751271616})&\num[scientific-notation=false,round-mode=places,round-precision=2]{0.8253183869683739} (\num[scientific-notation=false,round-mode=places,round-precision=2]{0.16708388304972935})\\
$\xi_1$&\num[scientific-notation=false,round-mode=places,round-precision=2]{0.0314327352404162} (\num[scientific-notation=true,round-mode=figures,round-precision=2,output-exponent-marker=\texttt{e}]{0.0018565080107486435})&\num[scientific-notation=false,round-mode=places,round-precision=2]{0.04260701303993951} (\num[scientific-notation=true,round-mode=figures,round-precision=2,output-exponent-marker=\texttt{e}]{0.008387984848445816})&\num[scientific-notation=false,round-mode=places,round-precision=2]{0.02025881791140559} (\num[scientific-notation=true,round-mode=figures,round-precision=2,output-exponent-marker=\texttt{e}]{0.002305710282458664})&\num[scientific-notation=true,round-mode=figures,round-precision=2,output-exponent-marker=\texttt{e}]{0.006805909471948028} (\num[scientific-notation=true,round-mode=figures,round-precision=2,output-exponent-marker=\texttt{e}]{0.0020298186824112274})&\num[scientific-notation=false,round-mode=places,round-precision=2]{0.062462804280392456} (\num[scientific-notation=true,round-mode=figures,round-precision=2,output-exponent-marker=\texttt{e}]{0.00561443762827014})&\num[scientific-notation=true,round-mode=figures,round-precision=2,output-exponent-marker=\texttt{e}]{0.0022364749791569144} (\num[scientific-notation=true,round-mode=figures,round-precision=2,output-exponent-marker=\texttt{e}]{0.0013617254695084497})\\
$\phi_1$&\num[scientific-notation=true,round-mode=figures,round-precision=2,output-exponent-marker=\texttt{e}]{106.55630990556314} (\num[scientific-notation=false,round-mode=places,round-precision=2]{17.669285010950023})&\num[scientific-notation=true,round-mode=figures,round-precision=2,output-exponent-marker=\texttt{e}]{163.1642968998744} (\num[scientific-notation=false,round-mode=places,round-precision=2]{17.63095818767095})&\num[scientific-notation=true,round-mode=figures,round-precision=2,output-exponent-marker=\texttt{e}]{225.68095846960998} (\num[scientific-notation=false,round-mode=places,round-precision=2]{48.7304904555172})&\num[scientific-notation=false,round-mode=places,round-precision=2]{4.0975264418465285} (\num[scientific-notation=false,round-mode=places,round-precision=2]{0.9542516102513577})&\num[scientific-notation=true,round-mode=figures,round-precision=2,output-exponent-marker=\texttt{e}]{107.68455072643069} (\num[scientific-notation=false,round-mode=places,round-precision=2]{28.438806242750694})&\num[scientific-notation=true,round-mode=figures,round-precision=2,output-exponent-marker=\texttt{e}]{362.1785813879047} (\num[scientific-notation=false,round-mode=places,round-precision=2]{45.747826330598706})\\
$\phi_2$&\num[scientific-notation=false,round-mode=places,round-precision=2]{0.6359725490737751} (\num[scientific-notation=false,round-mode=places,round-precision=2]{0.07299237803911589})&\num[scientific-notation=false,round-mode=places,round-precision=2]{1.2036340305829685} (\num[scientific-notation=false,round-mode=places,round-precision=2]{0.2727268060080421})&\num[scientific-notation=false,round-mode=places,round-precision=2]{1.033383962894429} (\num[scientific-notation=false,round-mode=places,round-precision=2]{0.32508943032980603})&\num[scientific-notation=false,round-mode=places,round-precision=2]{0.452868700374407} (\num[scientific-notation=false,round-mode=places,round-precision=2]{0.03503680295583386})&\num[scientific-notation=false,round-mode=places,round-precision=2]{0.5076179812105922} (\num[scientific-notation=false,round-mode=places,round-precision=2]{0.05846798765741223})&\num[scientific-notation=false,round-mode=places,round-precision=2]{1.1377237045025745} (\num[scientific-notation=false,round-mode=places,round-precision=2]{0.3935360195858968})\\
$\alpha$&\num[scientific-notation=true,round-mode=figures,round-precision=2,output-exponent-marker=\texttt{e}]{0.00021387614883648138} (\num[scientific-notation=true,round-mode=figures,round-precision=2,output-exponent-marker=\texttt{e}]{0.00030920949006195257})&\num[scientific-notation=true,round-mode=figures,round-precision=2,output-exponent-marker=\texttt{e}]{0.0001207429936049085} (\num[scientific-notation=true,round-mode=figures,round-precision=2,output-exponent-marker=\texttt{e}]{0.00013473074243283237})&\num[scientific-notation=true,round-mode=figures,round-precision=2,output-exponent-marker=\texttt{e}]{0.001827496398862378} (\num[scientific-notation=true,round-mode=figures,round-precision=2,output-exponent-marker=\texttt{e}]{0.0003366719947633128})&\num[scientific-notation=true,round-mode=figures,round-precision=2,output-exponent-marker=\texttt{e}]{0.0006923243909262534} (\num[scientific-notation=true,round-mode=figures,round-precision=2,output-exponent-marker=\texttt{e}]{0.0002410565476796101})&\num[scientific-notation=true,round-mode=figures,round-precision=2,output-exponent-marker=\texttt{e}]{0.00042088392609476504} (\num[scientific-notation=true,round-mode=figures,round-precision=2,output-exponent-marker=\texttt{e}]{0.00019492273176309306})&\num[scientific-notation=true,round-mode=figures,round-precision=2,output-exponent-marker=\texttt{e}]{1.0893738776132052e-5} (\num[scientific-notation=true,round-mode=figures,round-precision=2,output-exponent-marker=\texttt{e}]{0.0002658685601468546})\\
$\eta_{st}$&\num[scientific-notation=false,round-mode=places,round-precision=2]{0.03591593048261951} (\num[scientific-notation=true,round-mode=figures,round-precision=2,output-exponent-marker=\texttt{e}]{0.0003705187581149945})&\num[scientific-notation=false,round-mode=places,round-precision=2]{0.11060785908943425} (\num[scientific-notation=true,round-mode=figures,round-precision=2,output-exponent-marker=\texttt{e}]{0.0010405600227629877})&\num[scientific-notation=false,round-mode=places,round-precision=2]{0.02490159709312209} (\num[scientific-notation=true,round-mode=figures,round-precision=2,output-exponent-marker=\texttt{e}]{0.00022271884176356118})&\num[scientific-notation=false,round-mode=places,round-precision=2]{0.06675946271171274} (\num[scientific-notation=true,round-mode=figures,round-precision=2,output-exponent-marker=\texttt{e}]{0.0007281965183109152})&\num[scientific-notation=false,round-mode=places,round-precision=2]{0.05692452596723097} (\num[scientific-notation=true,round-mode=figures,round-precision=2,output-exponent-marker=\texttt{e}]{0.0006309083435799313})&\num[scientific-notation=false,round-mode=places,round-precision=2]{0.07107696891572479} (\num[scientific-notation=true,round-mode=figures,round-precision=2,output-exponent-marker=\texttt{e}]{0.0005526641131019762})\\
$\eta_{t}$&\num[scientific-notation=false,round-mode=places,round-precision=2]{0.013877115982555983} (\num[scientific-notation=true,round-mode=figures,round-precision=2,output-exponent-marker=\texttt{e}]{0.0006935848262811752})&\num[scientific-notation=false,round-mode=places,round-precision=2]{0.04857649820572244} (\num[scientific-notation=true,round-mode=figures,round-precision=2,output-exponent-marker=\texttt{e}]{0.0020078370520981203})&\num[scientific-notation=true,round-mode=figures,round-precision=2,output-exponent-marker=\texttt{e}]{0.007468919813436987} (\num[scientific-notation=true,round-mode=figures,round-precision=2,output-exponent-marker=\texttt{e}]{0.00043389086002070817})&\num[scientific-notation=false,round-mode=places,round-precision=2]{0.03128861154128622} (\num[scientific-notation=true,round-mode=figures,round-precision=2,output-exponent-marker=\texttt{e}]{0.0013959972649422083})&\num[scientific-notation=false,round-mode=places,round-precision=2]{0.02332918463458537} (\num[scientific-notation=true,round-mode=figures,round-precision=2,output-exponent-marker=\texttt{e}]{0.001103323531151892})&\num[scientific-notation=false,round-mode=places,round-precision=2]{0.028578073495440157} (\num[scientific-notation=true,round-mode=figures,round-precision=2,output-exponent-marker=\texttt{e}]{0.0011378629110166728})\\

   \hline
 \end{tabular}
 \end{adjustbox}
 \caption{Results from maximum likelihood estimation of the six days that were
 studied. Standard deviations obtained via the stochastic expected Fisher matrix
 are provided in parentheses.} 
 \label{tab:mles}
\end{table}

Table \ref{tab:mles} summarizes the point estimates for each
day and provides the implied standard deviation from the expected Fisher matrix.
Using this matrix as a proxy for the precision matrix of the MLE is a nontrivial
approximation, considering that the necessary conditions for that to even hold
asymptotically---for example, the log-likelihood looking quadratic near the
MLE---are likely not met, or at least are not met for this data size.
Nonetheless, as the expected Fisher matrix is the inverse of the Godambe
information matrix for the score equations, it still serves an at least somewhat
informative purpose about the uncertainty of each point estimate
\citep{heyde2008, stein2013}.

One main takeaway from these tables is that the parameters seem to be estimated
reasonably precisely and have believable values, although of interest is the
fact that for most days, there is at least one parameter that is not
well-resolved, and it is not always the same parameter. In some sense we find
this encouraging, as it is likely an unreasonable expectation that for a priori
domain choices every parameter is resolvable across any nontrivial collection of
days due to the domain variability discussed earlier. Robustness to this issue
seems to be a more important and realistic priority than full resolution at
every day and in all conditions. Further, the inconsistencies between which
parameters are not well-resolved serve as at least an indirect indication that
a model much simpler than the one used here would likely sacrifice some degrees
of freedom that maximum likelihood does consider to be valuable in at least some
circumstances.

Perhaps the most interesting of these estimates in its day-to-day variability is
the time-marginal smoothness-type parameter $\nu$. On June $20$th, $\nu$ is
quite low above the ABL, paradoxically indicating a rougher field
above the ABL height than below it. This is surprising, as from both a
heuristic application of meteorological principles and a visual inspection of
the data, one might expect $\nu$ above the ABL to be reasonably
large, and certainly larger than $\nu$ below the ABL. Moreover,
considering that there are two nugget-type terms, it is noteworthy that maximum
likelihood chose to represent the behavior of this field in $\nu$. This
contrasts with the estimates for June $03$, where maximum likelihood did
precisely the opposite: the estimated $\nu$ parameters are very large, but so are the
nugget parameters compared to the other days. The variability of this parameter
is particularly of interest considering that its continuous analog informs very
important properties of the process, such as differentiability.

\begin{figure}[!ht]
  \centering
  \scalebox{0.65}{\input{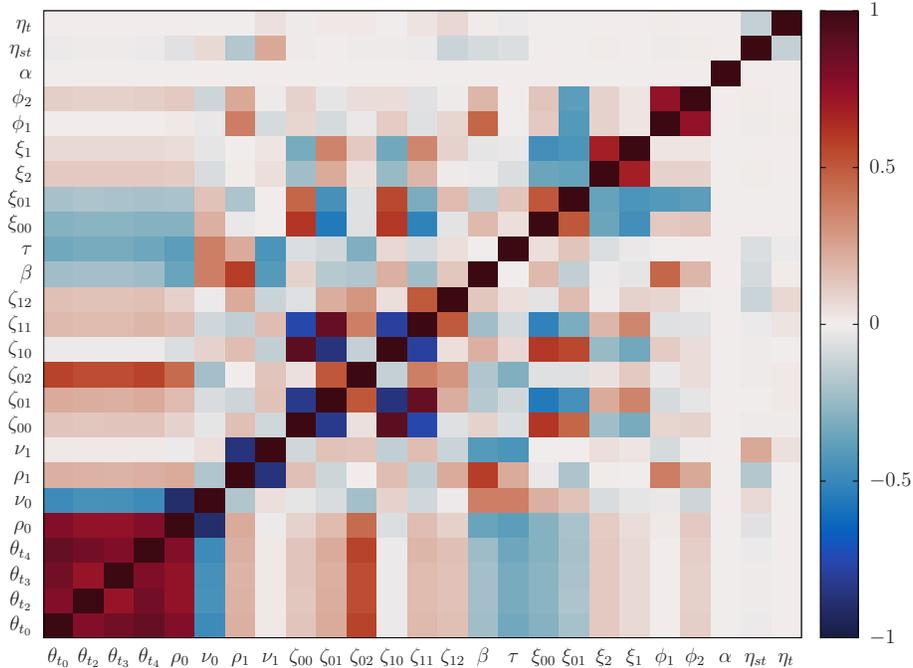}}
  \caption{The correlation matrix of the MLE of the model for the data on June $02$.}
  \label{fig:mle_02_corrhm}
\end{figure}

Several other parameters vary interestingly across days. The parameters
corresponding to marginal spectra do not vary in shocking ways, although the
$\set{\xi}$ parameters---which only inform additional very low frequency
power---appear to be reasonably sensitive to the particular structure of the
fifteen minute segments. While phase parameters often do not strongly influence
the likelihood, it is worth noting that for one of the six
days (June $06$) the MLE suggests a slight but well-estimated phase asymmetry.
Finally, we observe that the coherence parameters for the process above the ABL
are particularly interesting for June $03$. Specifically, the MLE suggests that
the coherence function above the ABL is almost flat at a value quite close to
$1$, suggesting a sort of deterministic process whose randomness comes largely
from the two nugget parameters. While it is obvious that coherence is high at low
frequencies above the ABL, it is noteworthy that for one of the
days---and only one of the days---maximum likelihood selected that type of
structure.

\begin{figure}%[!ht] 
  \centering
    \input{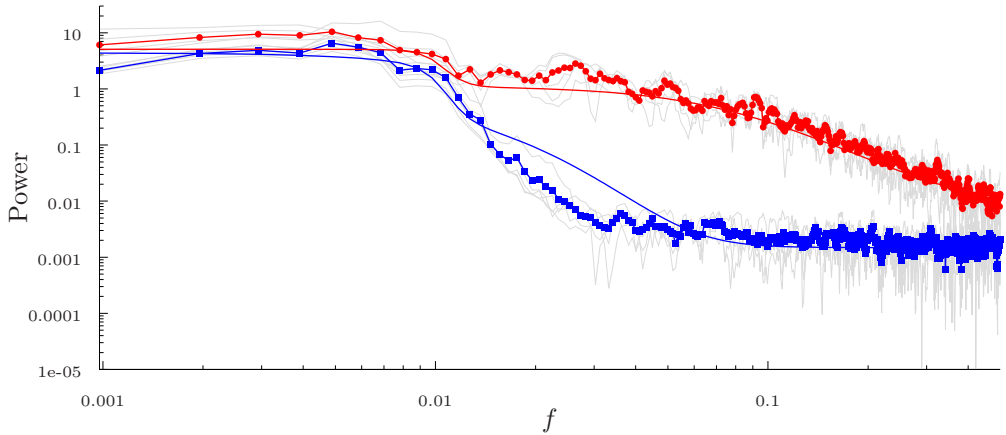}
    \caption{Marginal spectra for gates above the ABL (blue/squares) and
    below the ABL (red/circles) on June $02$.  Analytical estimates for these
    spectra (including nugget parameters) are overlaid as lines of corresponding
    colors. Estimators from individual range gates are shown faintly in grey.} 
    \label{fig:mle_specs}
\end{figure}

Figure \ref{fig:mle_02_corrhm} shows the correlation (not covariance) matrix for
the MLE on June $02$. As one might expect, the parameters determining the
effectively nonparametric scale function are highly correlated and are inversely
correlated with the smoothness, as is the range-type parameter $\rho$.  These
types of correlation structure are quite common in stationary Mat\'ern-type
models.  Moreover, there is strong negative correlation between the latter two
parameters used in the parameterization of $B$, which is the same phenomenon as
above. Of note, though, is that fixing one of those two parameters (or profiling
one in terms of the other) resulted in substantially lower likelihoods and much
more challenging optimization. Similarly, correlation of the $\set{\zeta_j(x)}$
parameters above and below the ABL is reasonably strong, likely due to the way
that they are combined for range gates near the mean ABL height, although that
is obviously another case where reducing the model would be wrong.  Otherwise,
correlations are not concerningly high in our assessment, and this further
suggests that the model is formulated well enough to avoid identifiability
problems despite the high parameter count.

Figures \ref{fig:mle_specs} and \ref{fig:mle_cohs} show the model-implied
marginal spectra and coherences for several specific distance respectively,
where coherences are shown in terms of their real and imaginary parts. The
marginal spectra in general seem well-captured by the model (shown here with the
nugget terms added), and the two-regime nonstationarity seems justifiable
after inspecting the individual gate spectra that are shown faintly in grey.
Maximum likelihood is well-known to prefer capturing high-frequency behavior
as accurately as possible over low-frequency behavior, and this observation in
an earlier stage of the model formulation was a primary motivator for the
extra low-frequency term $(1 + B(f))$ in (\ref{eq:sdf}), which costs a few extra parameters
but substantially improves the model's ability to capture the lowest frequency
behavior.  It can be seen in the modeled spectra above the ABL that the
implied smoothness doesn't agree as well as it could with the nonparametric
estimator.  In the supplemental information that shows every figure for all
six days, this is a noticeable pattern. To some degree, this is most likely
explained by the fact that the likelihood is presumably affected only weakly by
changing $S_x$ above the ABL in that frequency band and making that parameter
slightly incorrect in this marginal sense provided a greater improvement in
some other aspect of capturing covariance structure.

\begin{figure}%[!ht] 
  \centering
    \input{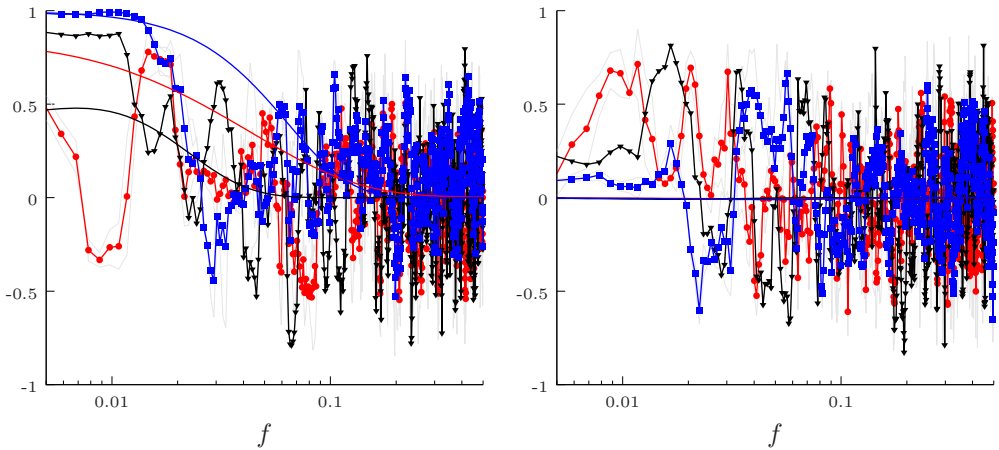}
    \caption{Average real (left) and imaginary (right) parts of the complex
      coherence between measurements three range gates apart above the ABL
      (blue/squares), below it (red/circles), and twelve gates apart across it
      (black/triangles), with individual estimators in faint grey.
    Individual estimates were obtained using multitaper estimators with five
    sine tapers.  Analytical estimates of each part from the MLE are overlaid in
    lines of the corresponding color.}
  \label{fig:mle_cohs}
\end{figure}

For measurements at this frequency and a reasonably small time segment,
coherences are unfortunately very difficult to estimate without either very
strong and potentially confounding bias or intractable variability. We attempt
to strike a balance between the two in Figure \ref{fig:mle_cohs}, although the
interpretation is still difficult. With frequency shown on the log scale, one
may be tempted to interpret the imaginary part of the coherence as meaningful,
it is quite variable and so it is unsurprising that the model in general
suggests that most of the energy of the coherence is in the real part. 

\begin{sidewaysfigure}
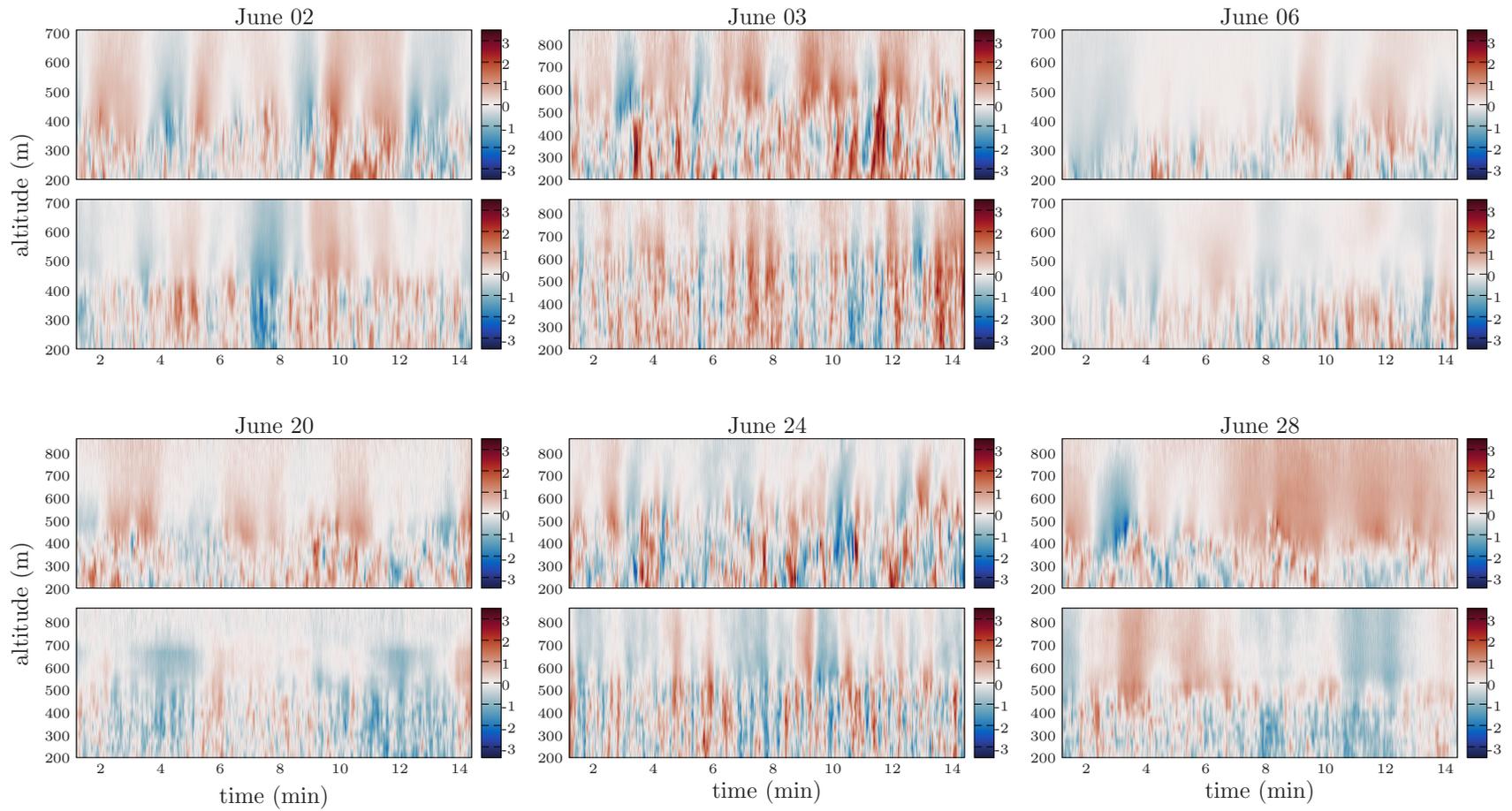
%[!ht] 
  \centering
  \begin{tabular}{ccc}
    \input{compare_02_hm.tex} & \input{compare_03_hm.tex} & \input{compare_06_hm.tex} \\
    \input{compare_20_hm.tex} & \input{compare_24_hm.tex} & \input{compare_28_hm.tex}
  \end{tabular}
  \caption{Side-by-side comparisons of the data used for estimation (above) and
  a simulation from the GP model with estimated parameters (below), all in units
  of m/s, for each fitted day.}
  \label{fig:alldays_sim_compare}
\end{sidewaysfigure}

Bearing in mind that the lowest frequencies are difficult to estimate and so the
nonparametric coherence estimators themselves may be suspect, the results
suggest that the model is not perfectly specified, in particular above the
ABL and in the form of the spatial dependence.  In general, the
coherence modeling here is more difficult than in \cite{stein2005b}, for
example, because of the very high rate of measurements. Noting
that these plots again show frequency on the log scale, we see that the coherence
estimates are not discernibly different from noise at even relatively low
frequencies, and so assessing the quality of the estimated coherence is
challening. A slightly more heuristic interpretation of the parametric
estimators indicates that the results are plausible, however, in that the MLE
does suggest near-unit coherence at low-frequencies above the ABL. 
While it is tempting to claim that a more complex parametric form might
represent these coherences better, inspection of similar figures for other days
in the supplementary material suggest that much of what one might interpret to
be structure in these estimates does not consistently appear across days, and a
more complex form for the coherence might result in overfitting that makes
multi-day extensions more challenging, rather than less so.

Finally, Figure \ref{fig:alldays_sim_compare} provides side-by-side
cross-section plots (reminiscent of Figure \ref{fig:data_hmplot}) of the
measurement data and a simulation from the fitted model for each of the six
days. While Table \ref{tab:mles} is valuable in assessing the quality of the
estimates, for a model of this complexity we also believe that assessing the
plausibility of simulations is important. In general, we find the simulations to
be satisfying in that regard, although several model mis-specifications are
uniquely visible in this format.  Interestingly, having a fixed estimator for
the ABL height $\beta$ and a smearing-type parameter $\tau$ does seem to
effectively create simulations where $\beta$ does not look fixed so long as
$\tau$ is relatively small. It does, however, highlight that a fixed $\beta$
means that the estimate will always be near the ABL's highest point in the data.
This makes sense considering that maximum likelihood will certainly penalize
under-estimation of uncertainty more harshly than over-estimation. (For example,
the likelihood of a $N(0, \sigma=100)$ variable being close to zero is much
larger than a $N(0,\sigma=1)$ being close to $100$). The extra variability below
the ABL means that the estimate for $\beta$ on June $28$th, for example, is off
by as much as $\SI{100}{\meter}$ from what one might obtain via more
sophisticated and targeted methods, for example those of \cite{tucker2009}.

In summary, this model does reasonably well in capturing both quantitative and
qualitative features of the process's covariance structure in the marginal
spectra. Moreover, visual assessment of the simulations suggests that at least
qualitative features of the coherence are reasonably well-captured. Considering
that the quality of the fits and simulations seems reasonably high across all
days, we find this to be generally encouraging. See the supplemental material
for analogs of all figures shown here for all six of the days studied.

\section{Discussion}

We have presented here a very flexible modeling framework for regularly measured
but nonstationary data. This method provides a convenient and powerful way to
think almost exclusively about single-dimension marginal properties of random
fields (like spectra and coherences) and obtain fully nonseparable, asymmetric,
and nonstationary covariance functions. We are not aware of other methods that
are sufficiently flexible that one can specify a $25$-parameter model that
composes two different and convenient representations of nonstationarity and for
which true maximum likelihood estimation can be performed. On a slightly less
significant note, we point out that this method is also uniquely capable of
providing estimates for the smoothness $\nu$ and its uncertainty. Derivatives of
the modified second-kind Bessel function $\mathcal{K}_{\nu}$ with respect to the
smoothness $\nu$ do exist, and there are series expansions that theoretically
facilitate computing them \citep{nist}, but we are not aware of any code to do
that, and in our own experimentation it is a daunting numerical challenge.  In
the Fourier domain, however, the derivatives are straightforward. Because of
this, the derivative of covariance matrices with respect to $\nu$ are obtained
completely analytically, and so there is no need to be concerned with derivative
approximations like finite differencing for the kernel function or covariance
matrix with respect to $\nu$.

With that said, however, this method of course has drawbacks and is somewhat
limited in the way that it can incorporate temporal nonstationarity (or, more
generally, nonstationarity in whatever marginal dimension is being transformed
over). Further, while some forms of irregularity in time can be mitigated with
proper padding and FFT-based interpolation, complete data irregularity---like
randomly sampled points---would make this covariance model much more onerous to
evaluate. In that sense, we see this method as being complementary to other
methods that are more robust in that respect, such as deformation methods
\citep{anderes2008} or convolution-based methods \citep{higdon2002,
paciorek2006}.

The data introduced here---and high-frequency vertical measurements of
meteorological processes in general, of which Doppler LIDAR measurements are but
one of many---provide a unique opportunity to study space-time random fields.
Unlike the well-studied Irish wind data \citep{haslett1989, gneiting2002,
stein2005, deluna2005, horrell2017}, this data source provides signed
measurements at such a high rate in both space and time that it is arguably a
different regime in both the relevant atmospheric physics and the relevant
statistical concerns.  The nonstationarity in altitude---a rarely-considered
spatial dimension that is as different from the other two as they all are from
time---is sufficiently sharp that there is strong dependence at two altitudes with
obviously different marginal behaviors, and thus provides an exceptional testing
ground for nonstationary models where they are the most necessary.

With regard to the estimation procedure, it merits consideration that the
analytical estimates of the marginal quantities via the MLE are quite different
than those obtained via approximate likelihoods, and almost necessarily will
look worse visually. In many spectral domain likelihood approximations, one is
effectively performing estimation via nonlinear least squares, which is a much
more direct attempt at curve fitting. Maximum likelihood, on the other hand,
will absolutely sacrifice the very low frequency representation of marginal spectra,
for example, if changing such a parameter would increase the likelihood at
higher spatio-temporal frequencies. As such, while one should of course expect
the estimated quantities $S_x$ and $C_f$ to be reasonably reflective of marginal
spectra and coherences, the estimation procedure itself is not attempting to
compartmentalize those quantities or provide the best visual fit of a parametric
curve to a nonparametric estimator.  In this sense, we find the visual agreement
of Figures \ref{fig:mle_specs} and \ref{fig:mle_cohs} to be satisfying, even if
a simpler curve-fitting estimation method could likely provide a more impressive
visualization.

A major objective of this work was to attempt to write a parametric model that
can do a reasonably good job of representing the covariance structure on
different days so that eventually those models may be stitched together to form
one that is cyclostationary in time on the timescale of weeks.  A myriad of
other concerns become relevant at that timescale; rain, clouds, and even
migratory birds and insect swarms \citep{muradyan2020} can introduce confounding
structure that this model is not adequately equipped to capture.  Nonetheless,
however, we are encouraged by what has been achieved here, and are hopeful that,
going forward, modifications for cyclostationarity will be made possible. One
notable concern from the estimation even across just June $2015$, however, is
that the day-to-day variation seems high enough that a purely periodic model for
many of the parameters (for example $\nu_1$) would likely miss a lot of
subtle features of the data that are arguably the most interesting, and perhaps
most informative about the micrometeorology of the area. A Bayesian perspective
may be more appropriate.

More generally, a primary objective for scientists when studying this data is to
provide an estimator for the ABL height, which we have denoted as $\beta$.  Many
works in the meteorological literature provide time series-based methods for
estimating this quantity---for example \cite{tucker2009}---but these methods
often do not fully exploit spatio-temporal structure, work on a coarser time
scale to provide estimates, and assimilate estimators from multiple data
sources.  While those estimators are obviously better than the estimate for
$\beta$ that is obtained via fitting the model described here, we believe that
this work demonstrates that it should be possible to bridge the gap between what
the meteorological community does to estimate this quantity and what the
spatio-temporal statistics community might do. A Bayesian approach that puts a
Mat\'ern-plus-drift prior on the time-varying process $ \beta(t)$, for example,
might reasonably provide meaningful posterior estimates for the ABL height on
the time scale of \emph{seconds}, as opposed to the time scale of $15$ minutes
in \cite{tucker2009}. Moreover, by virtue of being a fully spatio-temporal
model, such a procedure could exploit much more structure in the data while also
mitigating some potentially non-physical artifacts that occur from estimators
that primarily exploit only time dependence, such as large spatial
discontinuities. We believe that these and other similar questions are exigent
and will likely be fruitful topics of future work.

\renewcommand{\abstractname}{Acknowledgements}
\begin{abstract}
  \noindent The authors would like to thank Paytsar Muradyan for her invaluable
  support in understanding and working with the Doppler LIDAR data and for
  her guidance in the writing of Section $3$. They are also grateful to Lydia
  Zoells for her careful copy editing.

  \smallskip

  \noindent This material was based upon work supported by the U.S. Department
  of Energy, Office of Science, Office of Advanced Scientific Computing Research
  (ASCR) under Contracts DE-AC02-06CH11347 and DE-AC02-06CH11357. We acknowledge
  partial NSF funding through awards FP061151-01-PR and CNS-1545046 to MA. 
\end{abstract}

\bibliography{references}

@book{stein1999,
	title = {Interpolation of spatial data: some theory for kriging},
	isbn = {9780387986296},
	shorttitle = {Interpolation of spatial data},
	language = {en},
	publisher = {Springer Science \& Business Media},
	author = {Stein, Michael L.},
	month = jun,
	year = {1999}
}

@article{jun2008,
	title = {Nonstationary covariance models for global data},
	volume = {2},
	issn = {1932-6157, 1941-7330},
	doi = {10.1214/08-AOAS183},
	language = {EN},
	number = {4},
	journal = {Annals of Applied Statistics},
	author = {Jun, Mikyoung and Stein, Michael L.},
	month = dec,
	year = {2008},
	zmnumber = {1168.62381},
	pages = {1271--1289}
}

@techreport{stein2005,
	title = {Nonstationary spatial covariance functions},
	author = {Stein, Michael L.},
	year = {2005}
}

@article{guinness2013,
	title = {Interpolation of nonstationary high frequency spatial–temporal temperature data},
	volume = {7},
	issn = {1932-6157, 1941-7330},
	doi = {10.1214/13-AOAS633},
	language = {EN},
	number = {3},
	journal = {Annals of Applied Statistics},
	author = {Guinness, Joseph and Stein, Michael L.},
	month = sep,
	year = {2013},
	zmnumber = {06237193},
	pages = {1684--1708}
}

@article{poppick2016,
	title = {Temperatures in transient climates: {Improved} methods for simulations with evolving temporal covariances},
	volume = {10},
	issn = {1932-6157, 1941-7330},
	shorttitle = {Temperatures in transient climates},
	doi = {10.1214/16-AOAS903},
	language = {EN},
	number = {1},
	journal = {Annals of Applied Statistics},
	author = {Poppick, Andrew and McInerney, David J. and Moyer, Elisabeth J. and Stein, Michael L.},
	month = mar,
	year = {2016},
	zmnumber = {1358.62111},
	pages = {477--505}
}

@article{horrell2017,
	title = {Half-spectral space–time covariance models},
	volume = {19},
	issn = {2211-6753},
	doi = {10.1016/j.spasta.2016.12.002},
	language = {en},
	journal = {Spatial Statistics},
	author = {Horrell, Michael T. and Stein, Michael L.},
	month = feb,
	year = {2017},
	pages = {90--100}
}

@article{stein2005b,
	title = {Statistical methods for regular monitoring data},
	volume = {67},
	issn = {1467-9868},
	doi = {10.1111/j.1467-9868.2005.00520.x},
	language = {en},
	number = {5},
	journal = {Journal of the Royal Statistical Society: Series B (Statistical Methodology)},
	author = {Stein, Michael L.},
	year = {2005},
	pages = {667--687}
}

@article{haslett1989,
	title = {Space-time modelling with long-memory dependence: assessing ireland's wind power resource},
	volume = {38},
	issn = {1467-9876},
	shorttitle = {Space-time modelling with long-memory dependence},
	doi = {10.2307/2347679},
	language = {en},
	number = {1},
	journal = {Journal of the Royal Statistical Society: Series C (Applied Statistics)},
	author = {Haslett, John and Raftery, Adrian E.},
	year = {1989},
	pages = {1--21}
}

@article{gneiting2002,
	title = {Nonseparable, stationary covariance functions for space–time data},
	volume = {97},
	issn = {0162-1459},
	doi = {10.1198/016214502760047113},
	number = {458},
	journal = {Journal of the American Statistical Association},
	author = {Gneiting, Tilmann},
	month = jun,
	year = {2002},
	pages = {590--600}
}

@article{cressie1999,
	title = {Classes of nonseparable, spatio-temporal stationary covariance functions},
	volume = {94},
	issn = {0162-1459},
	doi = {10.1080/01621459.1999.10473885},
	number = {448},
	journal = {Journal of the American Statistical Association},
	author = {Cressie, Noel and Huang, Hsin-Cheng},
	month = dec,
	year = {1999},
	pages = {1330--1339}
}

@article{sampson1992,
	title = {Nonparametric estimation of nonstationary spatial covariance structure},
	volume = {87},
	issn = {0162-1459},
	doi = {10.1080/01621459.1992.10475181},
	number = {417},
	journal = {Journal of the American Statistical Association},
	author = {Sampson, Paul D. and Guttorp, Peter},
	month = mar,
	year = {1992},
	pages = {108--119}
}

@article{tucker2009,
	title = {Doppler lidar estimation of mixing height using turbulence, shear, and aerosol profiles},
	volume = {26},
	issn = {0739-0572},
	doi = {10.1175/2008JTECHA1157.1},
	number = {4},
	journal = {Journal of Atmospheric and Oceanic Technology},
	author = {Tucker, Sara C. and Senff, Christoph J. and Weickmann, Ann M. and Brewer, W. Alan and Banta, Robert M. and Sandberg, Scott P. and Law, Daniel C. and Hardesty, R. Michael},
	month = apr,
	year = {2009},
	pages = {673--688}
}

@book{ibragimov1978,
	address = {New York},
	series = {Stochastic {Modelling} and {Applied} {Probability}},
	title = {Gaussian random processes},
	isbn = {9780387903026},
	language = {en},
	publisher = {Springer-Verlag},
	author = {Ibragimov, I. A. and Rozanov, Y. A.},
	year = {1978},
	doi = {10.1007/978-1-4612-6275-6}
}

@article{paciorek2006,
	title = {Spatial modelling using a new class of nonstationary covariance functions},
	volume = {17},
	issn = {1099-095X},
	doi = {10.1002/env.785},
	language = {en},
	number = {5},
	journal = {Environmetrics},
	author = {Paciorek, Christopher J. and Schervish, Mark J.},
	year = {2006},
	pages = {483--506}
}

@article{anderes2008,
	title = {Estimating deformations of isotropic {Gaussian} random fields on the plane},
	volume = {36},
	issn = {0090-5364, 2168-8966},
	doi = {10.1214/009053607000000893},
	language = {EN},
	number = {2},
	journal = {Annals of Statistics},
	author = {Anderes, Ethan B. and Stein, Michael L.},
	month = apr,
	year = {2008},
	zmnumber = {1133.62077},
	pages = {719--741}
}

@techreport{fuentes2001b,
	title = {A new class of nonstationary models},
	institution = {Tech. report at North Carolina State University, Institute of Statistics},
	author = {Fuentes, M. and Smith, R.},
	year = {2001}
}

@article{cooley1965,
	title = {An algorithm for the machine calculation of complex {Fourier} series},
	volume = {19},
	issn = {0025-5718, 1088-6842},
	doi = {10.1090/S0025-5718-1965-0178586-1},
	language = {en},
	number = {90},
	journal = {Mathematics of Computation},
	author = {Cooley, James W. and Tukey, John W.},
	year = {1965},
	pages = {297--301}
}

@article{butterworth1930,
	title = {In the theory of filter amplifires},
	volume = {7},
	journal = {Experimental Wireless \& the Wireless Engineer},
	author = {Butterworth, S.},
	year = {1930},
	pages = {536--541}
}

@book{nocedal2006,
	title = {Numerical optimization},
	isbn = {9780387400655},
	language = {en},
	publisher = {Springer Science \& Business Media},
	author = {Nocedal, Jorge and Wright, Stephen},
	month = dec,
	year = {2006}
}

@article{geoga2019,
	title = {Scalable gaussian process computations using hierarchical matrices},
	volume = {0},
	issn = {1061-8600},
	doi = {10.1080/10618600.2019.1652616},
	number = {0},
	journal = {Journal of Computational and Graphical Statistics},
	author = {Geoga, Christopher J. and Anitescu, Mihai and Stein, Michael L.},
	month = aug,
	year = {2019},
	pages = {1--11}
}

@article{deluna2005,
	title = {Predictive spatio-temporal models for spatially sparse enviromental data},
	volume = {15},
	issn = {1017-0405},
	number = {2},
	journal = {Statistica Sinica},
	author = {de Luna, Xavier and Genton, Marc G.},
	year = {2005},
	pages = {547--568}
}

@inproceedings{higdon2002,
	address = {London},
	title = {Space and space-time modeling using process convolutions},
	isbn = {9781447106579},
	doi = {10.1007/978-1-4471-0657-9_2},
	language = {en},
	booktitle = {Quantitative {Methods} for {Current} {Environmental} {Issues}},
	publisher = {Springer},
	author = {Higdon, Dave},
	editor = {Anderson, Clive W. and Barnett, Vic and Chatwin, Philip C. and El-Shaarawi, Abdel H.},
	year = {2002},
	pages = {37--56}
}

@article{bezanson2017,
    title={Julia: A fresh approach to numerical computing},
    author={Bezanson, Jeff and Edelman, Alan and Karpinski, Stefan and Shah, Viral B},
    journal={SIAM {R}eview},
    volume={59},
    number={1},
    pages={65--98},
    year={2017},
    publisher={SIAM},
    doi={10.1137/141000671}
}

@article{zhang2004,
	title = {Inconsistent estimation and asymptotically equal interpolations in model-based geostatistics},
	volume = {99},
	issn = {0162-1459},
	doi = {10.1198/016214504000000241},
	number = {465},
	urldate = {2020-06-11},
	journal = {Journal of the American Statistical Association},
	author = {Zhang, Hao},
	month = mar,
	year = {2004},
	pages = {250--261}
}

@article{zhang2005,
	title = {Towards reconciling two asymptotic frameworks in spatial statistics},
	volume = {92},
	issn = {0006-3444},
	doi = {10.1093/biomet/92.4.921},
	language = {en},
	number = {4},
	urldate = {2020-06-11},
	journal = {Biometrika},
	author = {Zhang, Hao and Zimmerman, Dale L.},
	month = dec,
	year = {2005},
	pages = {921--936}
}

@book{heyde2008,
	title = {Quasi-likelihood and its application: a general approach to optimal parameter estimation},
	isbn = {9780387226798},
	shorttitle = {Quasi-likelihood and its application},
	language = {en},
	publisher = {Springer Science \& Business Media},
	author = {Heyde, Christopher C.},
	month = jan,
	year = {2008},
	note = {Google-Books-ID: forqBwAAQBAJ}
}

@article{stein2013,
	title = {Stochastic approximation of score functions for {Gaussian} processes},
	volume = {7},
	issn = {1932-6157, 1941-7330},
	doi = {10.1214/13-AOAS627},
	language = {EN},
	number = {2},
	urldate = {2020-06-11},
	journal = {Annals of Applied Statistics},
	author = {Stein, Michael L. and Chen, Jie and Anitescu, Mihai},
	month = jun,
	year = {2013},
	pages = {1162--1191}
}

@article{matsuo2011,
	title = {Nonstationary covariance modeling for incomplete data: {Monte} {Carlo} {EM} approach},
	volume = {55},
	issn = {0167-9473},
	shorttitle = {Nonstationary covariance modeling for incomplete data},
	doi = {10.1016/j.csda.2010.12.002},
	language = {en},
	number = {6},
	urldate = {2020-06-11},
	journal = {Computational Statistics \& Data Analysis},
	author = {Matsuo, Tomoko and Nychka, Douglas W. and Paul, Debashis},
	month = jun,
	year = {2011},
	pages = {2059--2073}
}

@article{nychka2002,
	title = {Multiresolution models for nonstationary spatial covariance functions},
	volume = {2},
	issn = {1471-082X},
	doi = {10.1191/1471082x02st037oa},
	language = {en},
	number = {4},
	urldate = {2020-06-11},
	journal = {Statistical Modelling},
	author = {Nychka, Douglas and Wikle, Christopher and Royle, J Andrew},
	month = dec,
	year = {2002},
	pages = {315--331}
}

@article{katzfuss2017,
	title = {A multi-resolution approximation for massive spatial datasets},
	volume = {112},
	issn = {0162-1459},
	doi = {10.1080/01621459.2015.1123632},
	number = {517},
	urldate = {2020-06-11},
	journal = {Journal of the American Statistical Association},
	author = {Katzfuss, Matthias},
	month = jan,
	year = {2017},
	pages = {201--214}
}

@article{katzfuss2012,
	title = {Bayesian hierarchical spatio-temporal smoothing for very large datasets},
	volume = {23},
	issn = {1099-095X},
	doi = {10.1002/env.1147},
	language = {en},
	number = {1},
	urldate = {2020-06-11},
	journal = {Environmetrics},
	author = {Katzfuss, Matthias and Cressie, Noel},
	year = {2012},
	pages = {94--107}
}

@article{genton2007,
	title = {Separable approximations of space-time covariance matrices},
	volume = {18},
	issn = {1099-095X},
	doi = {10.1002/env.854},
	language = {en},
	number = {7},
	urldate = {2020-06-11},
	journal = {Environmetrics},
	author = {Genton, Marc G.},
	year = {2007},
	pages = {681--695}
}

@article{genton2004,
	title = {On a time deformation reducing nonstationary stochastic processes to local stationarity},
	volume = {41},
	issn = {0021-9002, 1475-6072},
	doi = {10.1239/jap/1077134681},
	language = {en},
	number = {1},
	urldate = {2020-06-11},
	journal = {Journal of Applied Probability},
	author = {Genton, Marc G. and Perrin, Olivier},
	month = mar,
	year = {2004},
	pages = {236--249}
}

@book{nist,
  title={NIST handbook of mathematical functions hardback and CD-ROM},
  author={Olver, Frank WJ and Lozier, Daniel W and Boisvert, Ronald F and Clark, Charles W},
  year={2010},
  publisher={Cambridge university press}
}

@article{cushman2014,
  title={Atmospheric boundary layer},
  author={Cushman-Roisin, Benoit},
  journal={Environmental Fluid Mechanics},
  pages={165--186},
  year={2014}
}

@article{gage1978,
	title = {Doppler radar probing of the clear atmosphere},
	volume = {59},
	issn = {0003-0007},
	url = {https://journals.ametsoc.org/bams/article/59/9/1074/50676/Doppler-Radar-Probing-of-the-Clear-Atmosphere},
	doi = {10.1175/1520-0477(1978)059<1074:DRPOTC>2.0.CO;2},
	language = {en},
	number = {9},
	urldate = {2020-06-18},
	journal = {Bulletin of the American Meteorological Society},
	author = {Gage, K. S. and Balsley, B. B.},
	month = sep,
	year = {1978},
	pages = {1074--1094}
}

@article{mather2013,
	title = {The arm climate research facility: a review of structure and capabilities},
	volume = {94},
	issn = {0003-0007},
	shorttitle = {The arm climate research facility},
	doi = {10.1175/BAMS-D-11-00218.1},
	language = {en},
	number = {3},
	urldate = {2020-06-18},
	journal = {Bulletin of the American Meteorological Society},
	author = {Mather, James H. and Voyles, Jimmy W.},
	month = mar,
	year = {2013},
	pages = {377--392}
}

@book{muradyan2020,
	address = {United States},
	title = {Radar {Wind} {Profiler} ({RWP}) and {Radio} {Acoustic} {Sounding} {System} ({RASS}) {Instrument} {Handbook} ({DOE}/{SC}-{ARM}-{TR}-044)},
	url = {https://www.arm.gov/publications/tech\_reports/handbooks/rwp\_handbook.pdf},
	publisher = {DOE Office of Science Atmospheric Radiation Measurement (ARM) Program},
	author = {Muradyan, P and Coulter, R},
	year = {2020}
}

@techreport{newsom2012,
	title = {Doppler lidar ({Dl}) handbook},
	url = {https://www.osti.gov/biblio/1034640},
	language = {English},
	number = {DOE/SC-ARM/TR-101},
	urldate = {2020-06-18},
	institution = {DOE Office of Science Atmospheric Radiation Measurement (ARM) Program (United States)},
	author = {Newsom, R. K.},
	year = {2012},
	doi = {10.2172/1034640}
}

@article{sisterson2016,
	title = {The arm southern great plains ({Sgp}) site},
	volume = {57},
	issn = {0065-9401},
	doi = {10.1175/AMSMONOGRAPHS-D-16-0004.1},
	language = {en},
	urldate = {2020-06-18},
	journal = {Meteorological Monographs},
	author = {Sisterson, D. L. and Peppler, R. A. and Cress, T. S. and Lamb, P. J. and Turner, D. D.},
	year = {2016},
	pages = {6.1--6.14}
}

@article{stokes1994,
	title = {The atmospheric radiation measurement ({Arm}) program: programmatic background and design of the cloud and radiation test bed},
	volume = {75},
	issn = {0003-0007},
	shorttitle = {The atmospheric radiation measurement ({Arm}) program},
	doi = {10.1175/1520-0477(1994)075<1201:TARMPP>2.0.CO;2},
	language = {en},
	number = {7},
	urldate = {2020-06-18},
	journal = {Bulletin of the American Meteorological Society},
	author = {Stokes, Gerald M. and Schwartz, Stephen E.},
	year = {1994},
	pages = {1201--1222}
}

@techreport{Newsom2010,
	title = {Doppler {Lidar} ({DLFPT})},
	url = {https://www.archive.arm.gov/metadata/html/sgpdlfptS01.b1.html},
	institution = {Atmospheric Radiation Measurement (ARM) user facility},
	author = {Newsom, Robert and Krishnamurthy, Raghavendra},
  year = {2010},
	note = {doi: 10.5439/1025185}
}

\vspace{-0.15cm}
\begin{flushright}
  \scriptsize \framebox{\parbox{2.5in}{Government License: The submitted
      manuscript has been created by UChicago Argonne, LLC, Operator of Argonne
      National Laboratory (``Argonne").  Argonne, a U.S.  Department of Energy
      Office of Science laboratory, is operated under Contract No.
      DE-AC02-06CH11357.  The U.S. Government retains for itself, and others
      acting on its behalf, a paid-up nonexclusive, irrevocable worldwide
      license in said article to reproduce, prepare derivative works, distribute
      copies to the public, and perform publicly and display publicly, by or on
      behalf of the Government. The Department of Energy will provide public
      access to these results of federally sponsored research in accordance with
  the DOE Public Access Plan.
  http://energy.gov/downloads/doe-public-access-plan. }} 
  \normalsize
\end{flushright}

\newpage
\appendix

\section*{Supplementary Material}

\begin{figure}[!ht]
  \centering
  \input{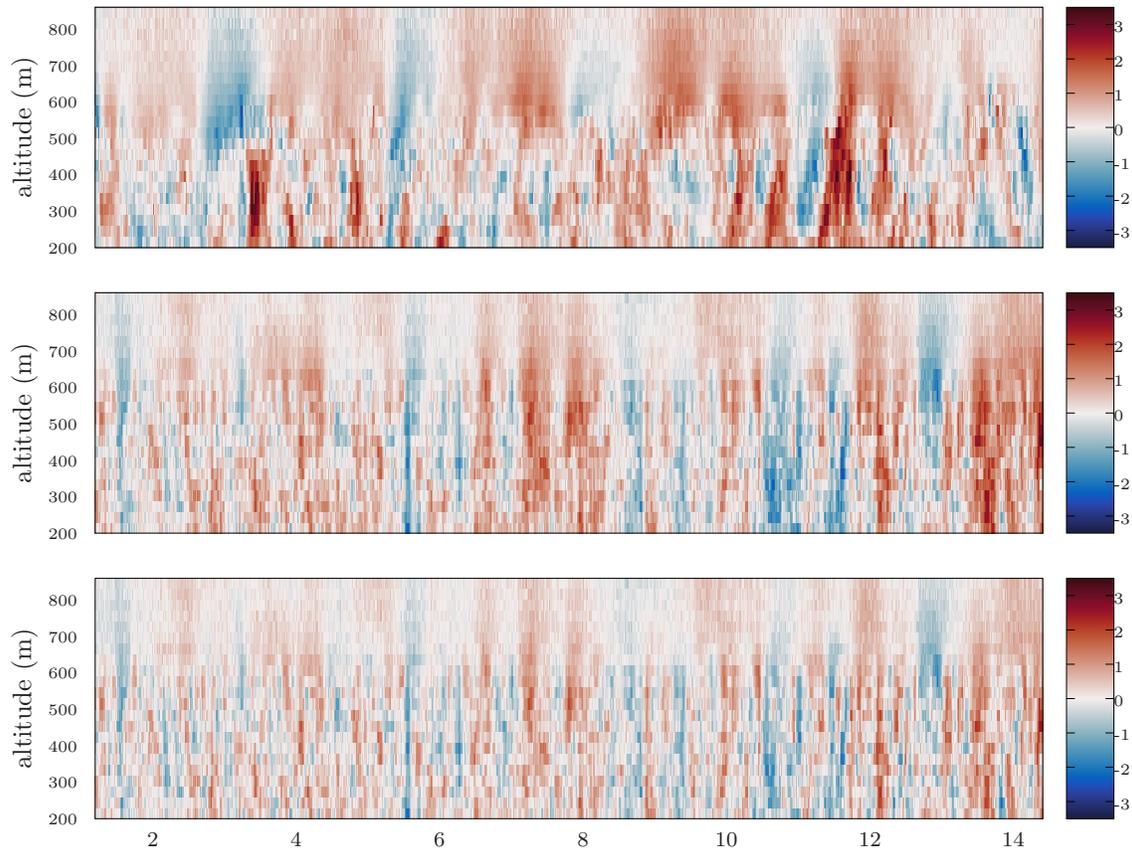}
  \caption{From top to bottom: The measurements used for fitting on June $03$, a
    simulation from the full model using fitted parameters, and a simulation
    from using fitted parameters of a subsetted model with the $\set{\xi}$
    parameters and phase parameter $\alpha$ removed, both using the same random
    seed. The difference in log-likelihood for all $\approx 17000$ points was only
    $183$ units, which is a very small per-observation difference. But as can be
    seen, the difference in visual quality of the simulations is substantial.}
  \label{fig:smallcomp}
\end{figure}

\newpage

\begin{figure}[!h] 
  \centering
    \scalebox{0.7}{\input{mle_02_14_corrhm.tex}}
    \input{mle_specs_02_14.tex}
    \input{mle_cohs_02_14.tex}
    \caption{Same as Figures \ref{fig:mle_02_corrhm}, \ref{fig:mle_specs}, 
    and \ref{fig:mle_cohs} respectively, for June $02$.}
\end{figure}

\newpage

\begin{figure}%[!ht] 
  \centering
    \scalebox{0.7}{\input{mle_03_14_corrhm.tex}}
    \input{mle_specs_03_14.tex}
    \input{mle_cohs_03_14.tex}
    \caption{Same as Figures \ref{fig:mle_02_corrhm}, \ref{fig:mle_specs}, and
    \ref{fig:mle_cohs} respectively, for June $03$.}
\end{figure}

\newpage

\begin{figure}%[!ht] 
  \centering
    \scalebox{0.7}{\input{mle_06_14_corrhm.tex}}
    \input{mle_specs_06_14.tex}
    \input{mle_cohs_06_14.tex}
    \caption{Same as Figures \ref{fig:mle_02_corrhm}, \ref{fig:mle_specs}, 
    and \ref{fig:mle_cohs} respectively, for June $06$.}
\end{figure}

\newpage

\begin{figure}%[!ht] 
  \centering
    \scalebox{0.7}{\input{mle_20_14_corrhm.tex}}
    \input{mle_specs_20_14.tex}
    \input{mle_cohs_20_14.tex}
    \caption{Same as Figures \ref{fig:mle_02_corrhm}, \ref{fig:mle_specs}, and
      \ref{fig:mle_cohs} respectively, for June $20$.}
\end{figure}

\newpage

\begin{figure}%[!ht] 
  \centering
    \scalebox{0.7}{\input{mle_24_14_corrhm.tex}}
    \input{mle_specs_24_14.tex}
    \input{mle_cohs_24_14.tex}
    \caption{Same as Figures \ref{fig:mle_02_corrhm}, \ref{fig:mle_specs}, and
      \ref{fig:mle_cohs} respectively, for June $24$.}
\end{figure}

\newpage

\begin{figure}%[!ht] 
  \centering
    \scalebox{0.7}{\input{mle_28_14_corrhm.tex}}
    \input{mle_specs_28_14.tex}
    \input{mle_cohs_28_14.tex}
    \caption{Same as Figures \ref{fig:mle_02_corrhm}, \ref{fig:mle_specs}, and
      \ref{fig:mle_cohs} respectively, for June $28$.}
\end{figure}

\end{document}